\documentclass[12pt]{report}
\usepackage{graphicx}
\usepackage{epstopdf}
\usepackage{amssymb,mathrsfs}
\DeclareMathAlphabet{\mathpzc}{OT1}{pzc}{m}{it}
\usepackage[lmargin=3.8cm,rmargin=2.6cm,tmargin=3cm,bmargin=3.5cm]{geometry}
\usepackage{nomencl}
\usepackage{amsmath,amssymb,amsfonts,amsthm}
\usepackage{algorithmic}
\usepackage{graphicx}
\usepackage{textcomp}
\usepackage{xcolor}
\usepackage{datetime}
\usepackage{lineno}
\newdateformat{monthyeardate}{\monthname[\THEMONTH], \THEYEAR}
\newtheorem{theorem}{Theorem}[chapter]
\newtheorem{lemma}{Lemma}[chapter]
\newtheorem{observation}{Observation}[chapter]
\newtheorem{definition}{Definition}[chapter]

\author{Name}
\date{}
\DeclareMathAlphabet{\mathpzc}{OT1}{pzc}{m}{it}
\begin{document}
\thispagestyle{empty}
\begin{large}
\begin{center}

{\bf Upper Bounds to Genome Rearrangement Problem using Prefix Transpositions}


\vspace{1.6cm}

A Thesis \\
\vspace{.2cm}
Submitted in partial fulfilment for the Degree of\\
\vspace{.2cm}
Doctor of Philosophy under the\\
\vspace{.2cm}
Faculty of Physical Sciences\\

\vspace{1.2cm}

by\\
\vspace{.1cm}
\textbf{PRAMOD P NAIR}\\

\vspace*{1.2cm}

\begin{figure}[h]
\center\includegraphics[width=12.58 cm]{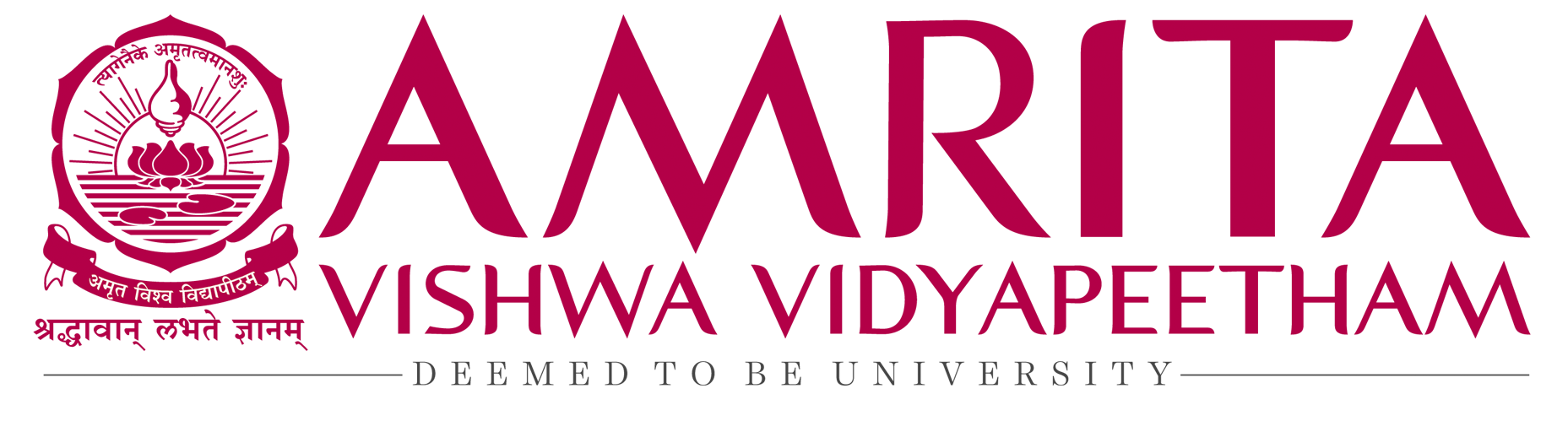}\\
\end{figure}

\vspace*{3cm}

\textbf{\small{ DEPARTMENT OF MATHEMATICS}}\\
\vspace*{.2cm}
\large {\textbf{AMRITA SCHOOL OF PHYSICAL SCIENCES}}\\
\vspace*{.2cm}
\Large{ \textbf{ AMRITA VISHWA VIDYAPEETHAM}}\\
\vspace*{.2cm}
\small{ AMRITAPURI} - 690 525 (\small{INDIA)}\\
\vspace*{.5cm}
\bf{\small{April, 2022}}
\end{center}
\end{large}

\baselineskip 1cm
\newtheorem{thm}{Theorem}[section]
\newtheorem{cor}{Corollary}[section]
\newtheorem{pro}{Proposition}[section]
\newtheorem{Lemma}{Lemma}[section]
\newtheorem{example}{Example}[section]
\newtheorem{rem}{Remark}[section]
\newtheorem{note}{Note}[section]
\newtheorem{defn}{Definition}[section]

\newpage
\thispagestyle{empty}

\begin{center}
{\large {\textbf{AMRITA SCHOOL OF PHYSICAL SCIENCES}}}\\
\vspace*{.18cm}
\Large{ \textbf{ AMRITA VISHWA VIDYAPEETHAM}}\\
\vspace*{.18cm}
\small{ AMRITAPURI} - 690 525\\

\vspace*{.6cm}

\begin{figure}[h]
\center\includegraphics*[width=12 cm]{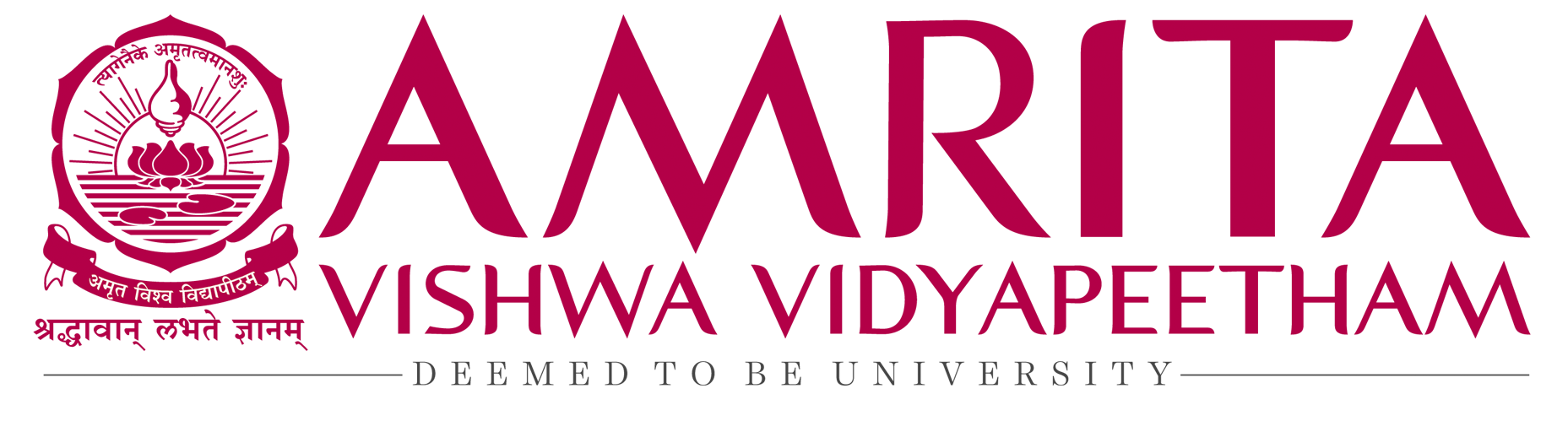}\\
\end{figure}

\vspace*{.6cm}

{\Large \textbf{BONAFIDE CERTIFICATE}}
\end{center}

\vspace*{.2cm}

\noindent This is to certify that the thesis entitled {\small\textbf{``Upper Bounds to Genome Rearrangement Problem using Prefix Transpositions"}} submitted by \textbf{Pramod P Nair, Register Number - AM.SC.D*MAT16223,} for the award of the \textbf{Degree of Doctor of Philosophy} under the \textbf{Faculty of Physical Sciences} is a bonafide record of the work carried out by him under my guidance and supervision at the Department of Mathematics, Amrita School of Physical Sciences, Amritapuri, Amritapuri.\\

\vspace*{2cm}

\baselineskip .5cm

\begin{center}
\textbf{Dr. Rajan Sundaravaradhan }\\
Thesis Advisor\\
Assistant Professor (Selection Grade)\\
Department of Mathematics\\
Amrita Vishwa Vidyapeetham, Amritapuri Campus\\
India\\
\end{center}

\newpage
\baselineskip 1cm
\thispagestyle{empty}

\begin{center}
\large {\textbf{AMRITA SCHOOL OF PHYSICAL SCIENCES}}\\
\vspace*{.2cm}
\Large{ \textbf{ AMRITA VISHWA VIDYAPEETHAM}}\\
\vspace*{.2cm}
\small{ AMRITAPURI} - 690525\\

\vspace*{1.3 cm}

\large {\textbf{DECLARATION}}

\end{center}

\vspace*{0.3cm}

\noindent{I, \textbf{Pramod P Nair, Register Number - AM.SC.D*MAT16223,} hereby declare that this thesis entitled \textbf{``Upper Bounds to Genome Rearrangement Problem using Prefix Transpositions",} is the record of the original work done by me under the guidance of \textbf{Dr. Rajan Sundaravaradhan,} Assistant Professor (Selection Grade), Department of Mathematics, Amrita Vishwa Vidyapeetham, Amritapuri. To the best of my knowledge, this work has not formed the basis for awarding any degree/diploma/ associateship/fellowship/or a similar award to any candidate in any University.}

\vspace*{1 cm}

\noindent{\textbf {Place:} Amritapuri \hfill \bf{Signature of the Student}}\\
\noindent{\textbf {Date:} 25.04.2022}

\vspace*{.4cm}

\begin{center}
\large {COUNTERSIGNED}\\

\vspace*{1.3cm}

\textbf{Dr. Rajan Sundaravaradhan }\\
Thesis Advisor\\
Assistant Professor (Selection Grade)\\
Department of Mathematics\\
Amrita Vishwa Vidyapeetham, Amritapuri Campus\\

\end{center}

\newpage
\baselineskip 1cm
\pagenumbering{roman}
\tableofcontents
\newpage

\section*{Acknowledgement}\addcontentsline{toc}{chapter}{Acknowledgement}
I am greatly indebted to my thesis advisor Dr. Rajan Sundaravaradhan, Assistant Professor (Selection Grade), Department of Mathematics, Amrita School of Physical Sciences, Amritapuri, for accepting me as his PhD student. His invaluable supervision, motivation, support, and tutelage during this period helped me select and work on this topic. I would like to express my wholehearted gratitude to him throughout my lifetime.

Dr. Bhadrachalam Chitturi, Associate Professor of Instruction, Department of Computer Science, University of Texas at Dallas, is an expert in my field of research. His knowledge and motivation deeply inspire me. I take this opportunity to thank him for his knowledgeable advice and meticulous scrutiny while publishing my research and writing the thesis.

I thank my doctoral committee members Dr. Ganesh Sundaram, Chairperson of Physics Department at Amritapuri Campus, Dr. Rajmohan Kombiyil, Assistant Professor at the Department of Physics, Amritapuri Campus; Dr. Georg Gutjahr, Assistant Professor at Amrita CREATE and Dr. Kurunandan Jain, Assistant Professor at Center for Cyber Security Systems and Networks, and for their insightful comments and encouragement.

I thank Dr. Narayanankutty Karuppath, Principal, School of Physical Sciences, Amritapuri, Dr. Ushakumari P V, Chairperson of Mathematics Department, and all the faculty members at the Department of Mathematics, Amrita School of Physical Sciences, Amritapuri for their support throughout this endeavour.

I would like to express my gratitude to my wife Soumya and my son Harshal. Without their unfailing love, tremendous understanding, and constant encouragement during the past few years, it would not have been possible for me to complete my study.

My parents, who imbibed moral values and social responsibilities in me, and my teachers, who instilled analytical thinking and professionalism, need a special mention. I thank them for what I am today.

Above all, I express my indebtedness to our beloved AMMA, Her Holiness Mata
Amritanandamayi Devi, Chancellor of Amrita Vishwa Vidyapeetham, for her bountiful
love and blessings, which helped me take up research at this university.


\baselineskip .7cm

\cleardoublepage
\addcontentsline{toc}{chapter}{\listfigurename}
\listoffigures

\cleardoublepage



\section*{{\huge List of Symbols}\hfill} \addcontentsline{toc}{chapter}{List of Symbols}
\vspace{1.2cm}
\begin{tabular}{|c|l|}
\hline
Symbols & Description\\
\hline
$\pi$ & permutation with $n$ symbols\\
$S_n$ & symmetric group with $n$ symbols\\
$I_n$ & identity permutation $(0,1,2,\dots, (n-2),(n-1))$\\
$R_n$ & reverse permutation $((n-1),(n-2),\dots, 2,1,0)$\\
$t_i$ & $i^\text{th}$ visited symbol in sequence length algorithm\\
$s_i$ & skipped symbol (if exists) that precedes $t_i$ in sequence length algorithm\\
$C$ & block of a permutation\\
\hline
\end{tabular}


\newpage

\chapter*{Abstract\hfill} \addcontentsline{toc}{chapter}{Abstract}

A Genome rearrangement problem studies large-scale mutations on a set of DNAs in living organisms. Various rearrangements like reversals, transpositions, translocations, fissions, fusions, and combinations and different variations have been studied extensively by computational biologists and computer scientists over the past four decades. From a mathematical point of view, a genome is represented by a permutation. The genome rearrangement problem is interpreted as a problem that transforms one permutation into another in a minimum number of moves under certain constraints depending on the chosen rearrangements. Finding the minimum number of moves is equivalent to sorting the permutation with the given rearrangement. A transposition is an operation on a permutation that moves a sublist of a permutation to a different position in the same permutation. A \emph{Prefix Transposition}, as the name suggests, is a transposition that moves a sublist which is a prefix of the permutation.

In this thesis, we study prefix transpositions on permutations and present a better upper bound for sorting permutations with prefix transpositions. A greedy algorithm called the \emph{generalised sequence length algorithm} is defined as an extension of the sequence length algorithm where suitable alternate moves are also considered. This algorithm is used to sequentially improve the upper bound to $n-\log_{3.3} n$ and $n-\log_3 n$. In the latter part of the thesis, we defined the concept of a \emph{block}. We used it along with the greedy moves of the generalised sequence length algorithm to get an upper bound of $n-\log_2 n$ to sort permutations by prefix transpositions.

\newpage
\baselineskip1cm
\pagenumbering{arabic}
\setcounter{page}{1}
\chapter{Introduction}
\section{Preliminaries}
A permutation $\pi$ on a set $S$ is an arrangement of the elements of $S$ in some order. Mathematically, a permutation is defined as a bijection on the set $S$. The collection of all permutations on a set with $n$ elements forms a group with operation composition. This group is called the symmetric group and is denoted as $S_n$. Permutations are used in many areas of pure sciences, medicine, and engineering. Due to the increased use of information transfers through the internet in the recent past, communication networks, cryptography, and network security systems use different permutations for performance evaluation and secure transfer. Permutations play a significant role in computer science in designing computer chips, data mining, and pattern analysis. Permutations also have applications in quantum physics for describing states of particles and computational biology for sequencing problems involving atoms, molecules, DNAs, genes, and proteins.

Genomes in a cell are mathematically modelled as permutations or strings. Operations like insertion, deletion, or substitution of a character in a string are called local operations. A global operation changes a substring in the string. A variety of local and global operations are defined in the literature. When all the characters in a string are distinct, we consider it a permutation. The distance between two permutations $\alpha$ and $\beta$ is defined as the number of operations required to transform $\alpha$ to $\beta$. This thesis studies a specific type of global operation called the prefix transposition and improves the upper bound to find the distance between any two permutations on $S_n$ using prefix transpositions.

\section{Motivation and Background} 
DNA sequencing is the process of finding the order of bases in the nucleotides of a cell. DNA sequencing techniques were initially developed in 1970. A well-known sequencing method was developed by Frederick Sanger \cite{sanger1977dna} in 1977. His sequencing method was used for the Human Genome Project, which determined the DNA sequence of the entire euchromatic human genome in 13 years. During this project, researchers sequenced all the 3.2 billion base pairs in the human genome. Since then, DNA sequencing has gained a lot of interest among researchers due to its impact on the study of evolution theory, medical diagnosis of genetic diseases, anthropology, and forensics. By 2000, many new DNA sequencing techniques were developed and implemented for commercial use. These are collectively called the next-generation sequencing (NGS) methods. NGSs fragment the DNA to the order of millions and sequence them at a cheaper and faster rate compared to the first-generation techniques. A variety of algorithms to analyse and interpret the enormous amount of data accumulated and sequenced by NGS has been proposed in the literature with improved computational complexities \cite{Li2006}.

A \emph{genome} is the set of all the chromosomes present in the nucleus of a cell of living organisms. Each of these chromosomes consists of DNA, which comprises two long sequences of nucleotides held together by hydrogen bonding. The nucleotides in each of the sequences are interrelated in a specific manner with four bases, namely - adenine(A), guanine(G), cytosine(C), thymine(T). The DNA replicates itself at a rate of about a thousand nucleotides per second to produce one genome from the other that is almost identical with a minimal level of inaccuracy. Though the rate of inaccuracy is only about one in $10^5$ to $10^9$ replications, it has been the root cause in the study of the theory of evolution, identification of genetic or hereditary diseases, and diagnosis of various characteristics in cultivated plants or bred animals. Dobzhansky and Sturtevant \cite{Dobzhansky1938}  published the first research article on genome rearrangement for molecular evolution in 1938. A review on the genome rearrangements in inherited diseases and cancer was published by Jian-Min Chen et al. \cite{chen2010genomic} in 2010. Kantar \cite{kantar2017genetics} has provided a summary on plant domestication due to genome sequencing. A detailed review of the technologies and applications in DNA sequencing over the past 40 years was presented by Shendure \cite{shendure2017dna} in 2017.

An alteration or change in the nucleotide sequence of a DNA or genome is called a \emph{mutation}. A new sequence that evolves by mutations at the level of nucleotides is called point
mutations. Mutations that occur at the level of chromosomes or genomes are called large-scale mutations. In a point mutation, the base of a nucleotide is either deleted, substituted by another base, or a new base is inserted. These types of mutations are detected and studied as \emph{sequence alignments} \cite{Jones2004}. Other sequences that evolve because of large-scale mutations are termed \emph{rearrangements}. A genome rearrangement usually occurs when a chromosome breaks at two or more locations (called the breakpoints), and the pieces are reassembled in a different order. In this case, we may also have deletion or insertion of a large number of chromosomes or movement of a section of the DNA to a different location. There are different types of rearrangements available in the literature, among which different variations of reversals, transpositions, translocations, fusion, fission, and their combinations are commonly encountered \cite{Li2006}. Reversals and Transpositions are rearrangements that are restricted to a single chromosome. In a reversal, a segment of the chromosome is reversed. Transpositions remove a segment of the chromosome and insert it in a different location. When a segment of one chromosome is exchanged with another chromosome, we get translocations. Fusion and fission occur when two chromosomes are joined into one or a chromosome is split into two. A pictorial representation of some common rearrangements is shown in figure \ref{types of GR}. 

\begin{figure}[h]
\centering
\includegraphics[scale=.5]{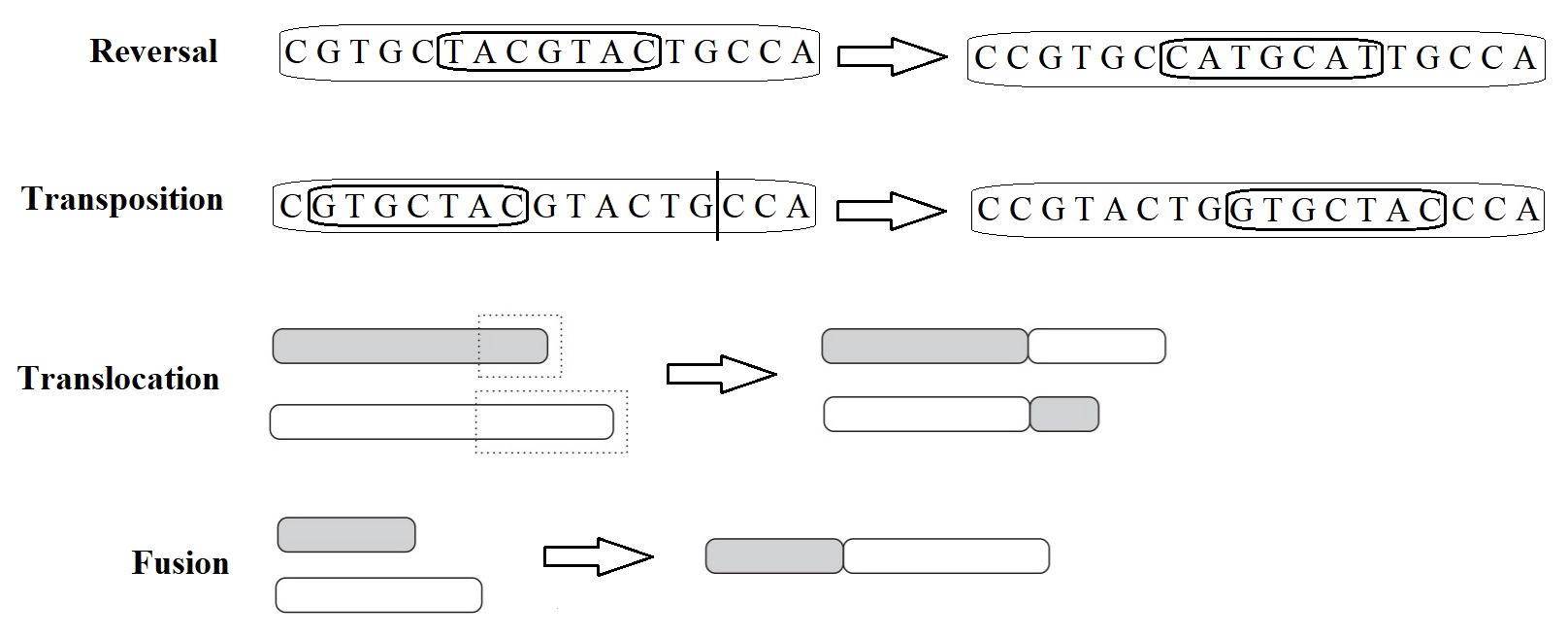}
\caption{Different types of genome rearrangements}
\label{types of GR}
\end{figure}

Detecting these rearrangements in different living organisms is termed as the \emph{genome rearrangement problem}. Nowadays, this term is used in a much broader sense and includes problems that are not directly motivated by molecular biology. For example, a block interchange that swaps two segments of a chromosome was defined and studied by Christie in his PhD thesis \cite{Christie1998} as a generalisation of transposition in 1998. It was later in 2005 that Lin \cite{lin2005efficient} showed that block-interchanges play a significant role in the evolution of Vibrio species. Of the several variants of genome rearrangement problems, the ‘prefix’ constraint, where the rearrangements are made on the prefix of a permutation, is the most common. Besides genome rearrangements, it finds practical applications in interconnection network design to reduce the value of the network's diameter and thus minimise the size of the network being generated. This method helps in reducing the maximum value for the delay in communications \cite{lakshmivarahan1993symmetry}. Prefix reversals \cite{gates1979bounds} and prefix transpositions \cite{Dias2002} have been studied extensively. Recently Labarre \cite{labarre2020sorting} studied sorting of prefix block-interchanges and provided a 2- approximation algorithm for the problem.

\section{Literature survey}
In 1982, Watterson \cite{Watterson1982} represented the position of genes in a genome by a permutation where a gene is represented as a symbol in a set and the chromosome is a permutation on the set. In his research article Watterson \cite{Watterson1982} considered the reversal distance problem for circular permutations, where the first gene is considered to be adjacent to the last gene. In this thesis, we denote a permutation $\pi$ over set $S={\{0,1,2,\dots,(n-1)\}}$ with $n$ symbols as $\pi=(\pi_1,\pi_2,\dots,\pi_n)$, where each $\pi_i \in S$. $(\pi_i,\pi_{i+1},\dots,\pi_j)$ is called a sublist of permutation $\pi$, where $1\le i\le j\le n$. The proposed genome rearrangement problem is to transform one permutation into another in a minimum number of operations under certain constraints depending on the rearrangements being considered. Transforming any permutation $\pi^{\star}\in S_n$ into any permutation $\pi^{\#}\in S_n$ is equivalent to sorting some $\pi'= ((\pi^{\#})^{-1})\pi^{\star}\in S_n$ \cite{Akers1989, lakshmivarahan1993symmetry}, i.e. transforming $\pi'$ to the identity permutation $I_n=(0,1,2,\dots,(n-1))$; here $\pi^{\star}$ is applied to the inverse of $\pi^{\#}$. Thus, in literature, the genome rearrangement problem is considered a sorting problem of permutations due to this property. The most common and widely studied genome rearrangements are reversals \cite{Watterson1982} and transpositions \cite{Bafna1998}. Another variant of these operations that allow combinations and weights for reversal and transpositions have been studied recently in \cite{Oliveira2019, alexandrino2020complexity}.

\subsection{Reversals}
A reversal operation $\alpha (i,j)$ on permutation $\pi=(\pi_1,\pi_2,\dots,\pi_n)$ is an operation that reverses the order of sublist $(\pi_i,\pi_{i+1},\ldots,\pi_j)$ in $\pi$. The reversal $\alpha (i,j)$ is denoted as

\begin{equation*}
\begin{gathered}
(\pi_1,\pi_2,\ldots,\pi_{i-1}[\pi_i,\pi_{i+1},\ldots,\pi_{j-1},\pi_j],\pi_{j+1},\ldots,\pi_n)\\
\rightarrow (\pi_1,\pi_2,\ldots,\pi_{i-1}[\pi_j,\pi_{j-1},\ldots,\pi_{i+1},\pi_i]\pi_{j+1},\ldots,\pi_n)
\end{gathered}
\end{equation*}
A prefix reversal is a particular case of reversal where $i=1$. If $\pi$ is a signed permutation, in addition to reversing the order of the sublist, the sign of each symbol in the sublist also changes. Given two permutations $\pi$ and $\pi^{\star}$ in $S_n$, the reversal distance between them is the minimum number of reversals required to transform $\pi$ to $\pi^{\star}$.
Even though finding the reversal distance of permutations was among the first type of sorting for genome arrangement problems, Kececioglu and Sankoff \cite{kececioglu1995exact}  published a significant result for sorting permutations by reversals in 1995. This paper provided the first approximation algorithm for sorting by reversals with a factor of 2 and identified some open problems related to chromosome rearrangements. Since then, reversals have been considered to be a core rearrangement in the field of computational biology. A variety of algorithms for sorting permutation by different types of reversals have been developed, of which prefix reversals \cite{chitturi200918}, signed reversals \cite{kaplan2000faster}, and some restricted reversals \cite{bender2008improved} have been widely studied in terms of their complexity and bounds. Some common reversals are shown in figure \ref{types of rev}.

\begin{figure}[h]
\centering
\includegraphics[scale=.7]{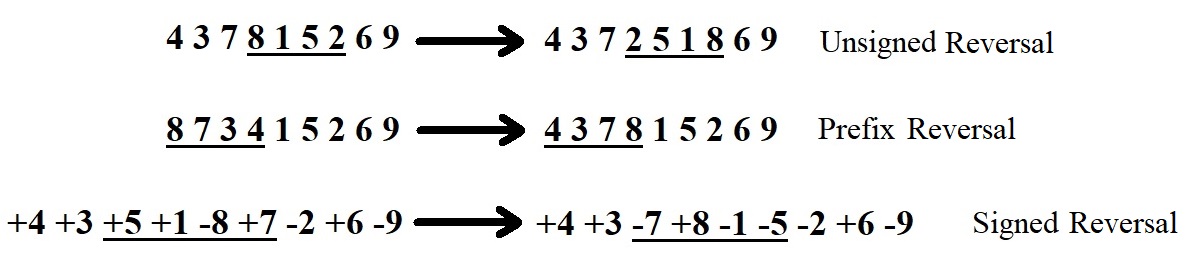}
\caption{Reversals}
\label{types of rev}
\end{figure}

If the orientation of the genes is considered, we use signed permutations to represent the genome. In this case, the reversal changes the sign of the gene while reversing the order. Signed permutations can be sorted using reversals in polynomial time \cite{hannenhalli1999transforming}, while Alberto Caprara \cite{Caprara1997} showed that the problem of sorting unsigned permutations by reversals is NP-hard. The best-known algorithm with a factor of 1.375 was proposed for this problem by Berman et al. \cite{berman20021}.

\subsection{Transpositions} 
A transposition $\alpha (i,j,k)$ is an operation on the symmetric group $S_n$ that transforms $\pi=(\pi_1,\pi_2,\dots,\pi_n)$ into another permutation by moving the sublist $(\pi_i,\pi_{i+1},\dots,\pi_{j-1})$ to the position between $\pi_{k-1}$ and $\pi_k$. We denote the transposition $\alpha (i,j,k)$ as 

\begin{equation*}
\begin{gathered}
(\pi_1,\pi_2,\dots,\pi_{i-1},[\pi_i,\dots,\pi_{j-1}],\pi_j,\dots,\pi_{k-1}* \pi_k,\dots,\pi_n) \\
\rightarrow (\pi_1,\pi_2,\dots,\pi_{i-1},\pi_j,\dots,\pi_{k-1},\pi_i,\dots,\pi_{j-1},\pi_k,\dots,\pi_n)\\
\end{gathered}
\end{equation*}
\\where the moved sublist is enclosed in parentheses and the destination position is marked with an asterisk \cite{Chitturi2015}. Given two permutations $\pi^{\#}$ and $\pi^{\star}$ in $S_n$, the transposition distance between them is the minimum number of transpositions required to transform $\pi^{\#}$ to $\pi^{\star}$. Bafna and Pevzner studied transpositions \cite{Bafna1998} and gave a $\frac{3}{2}$ approximation algorithm for sorting permutations by transpositions in 1998. They also provided a lower bound of $\lfloor{\frac{n}{2}\rfloor}+1$ and an upper bound of $\frac{3n}{4}$ for the transposition distance between two permutations. The upper bound was further improved to $\lfloor{\frac{2n-2}{3}\rfloor}$ by Eriksson et al. \cite{eriksson2001sorting} in 2001. This article also showed that the reverse order permutation $R_n=((n-1),(n-2),\dots,2,1,0)$ can be sorted in $\lceil{\frac{n+1}{2}\rceil}$. The best-known algorithm with a factor of 1.375, to sort permutations by transpositions, was proposed by Elias and Hartman \cite{Elias2006}. Transpositions have been studied extensively and several variations were explored. A prefix transposition is a special case of transposition where $i=1$ where the sublist of symbols up to $\pi_{j-1}$ is moved to the position before $\pi_k$ in $\pi$. An inverted transposition is a combination of transposition operation followed by reversal on the moved sublist. Heath and Vergara studied sorting permutations by bounded block moves \cite{heath1998sorting} whereas Feng et al. studied sorting permutations with cyclic adjacent transpositions \cite{Feng2010}. Examples of various transpositions and its generalization called block interchange is shown in figure \ref{types of trans}. The complexity class of sorting permutations by transpositions was proved to be NP-hard by giving a polynomial time reduction from SAT in 2012 \cite{Bulteau2012}.

\begin{figure}[h]
\centering
\includegraphics[scale=.7]{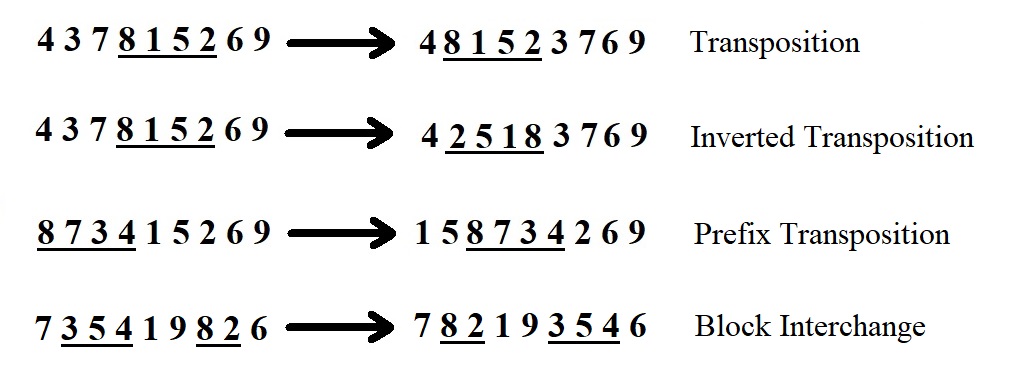}
\caption{Transpositions and Block interchange}
\label{types of trans}
\end{figure}

\section{Problem under investigation}
Transposable elements are segments of DNA that can move positions within a genome. They are present in all living organisms that are examined. The recombination of these elements in the genome results in various types of transpositions. Hence transpositions are considered to be one of the most common mutations that occur in living organisms. Prefix transpositions were discussed and studied initially in 2002 by Dias and Meidanis \cite{Dias2002} as a variation of the transposition problem. They provided lower and upper bounds of $\frac{n}{2}$ and $n-1$, respectively, using a 2-approximation algorithm for the problem. They also presented an algorithm that sorted the reverse permutation $R_n$ in $\frac{3n}{4}$ moves. It has been conjectured that a permutation with $n$ symbols can be sorted in $\frac{3n}{4}$ prefix transpositions \cite{Chitturi2008}. Over the past two decades, the upper and lower bounds to sort permutations by prefix transpositions have been improved to $n-\log_{(\frac{7}{2})}n$ \cite{Chitturi2015} and $\frac{3n}{4}$ \cite{labarre2008edit} respectively. We see that the lower bound has achieved the conjectured value, but there is a wide gap in terms of the upper bound. Hence, in this thesis, we introduce the generalised sequence length algorithm and blocks of a permutation to reduce the gap and provide a series of improvements on the upper bound.

In Chapter 2, we shall discuss prefix transpositions on permutations in detail and review the procedures that have been adopted to improve the upper bound to sort permutations by prefix transpositions from $n-1$ to $n-\log_{(\frac{7}{2})}n$. 

In Chapter 3, we provide an $n-\log_{(\frac{10}{3})}n$ upper bound to sort permutations by prefix transpositions which is a slight improvement to the previous best upper bound. Here we use the sequence length algorithm described in \cite{Chitturi2012} along with some additional prefix transpositions (called alternate moves).

In Chapter 4, we improve the upper bound further to $n-\log_{3}n$ using a technique similar to that in Chapter 3 by defining more alternate moves in different scenarios. We see that this method can be further implemented to improve the upper bound to at most $n-\log_{(1+\epsilon)}n$, but in doing so, the number of alternate moves that we need to find would be very high and thus make the proof very lengthy and complicated.

In Chapter 5, we propose the concept of a \emph{block} in a permutation and use it along with the generalised sequence length algorithm to improve the upper bound to sort a permutation with prefix transposition to $n-\log_{2}n$. The proof approach is entirely different from Chapters 3 and 4.

Conclusion and suggestions for further improvement of the upper bound are included in Chapter 6.

\chapter{Prefix Transpositions}

\section{Introduction}
A prefix transposition is an operation defined on the symmetric group $S_n$ that moves a sublist containing the first symbol of a permutation to a different position in the permutation. Hence it is a special case of transposition $\alpha (i,j,k)$ with $i=1$. We denote a prefix transposition as $\alpha (i,j), 1\le i<j \le n$, where the sublist $(\pi_1,\pi_2,\dots,\pi_{i-1})$ of the permutation $\pi=(\pi_1,\pi_2,\dots,\pi_n)$ is moved to a position between $\pi_{j-1}$ and $\pi_j$.

\begin{equation*}
\begin{gathered}
([\pi_1,\pi_2,\dots,\pi_{i-1}],\pi_i,\dots,\pi_{j-1}* \pi_j,\dots,\pi_n) \\
\rightarrow (\pi_i,\dots,\pi_{j-1},\pi_1,\pi_2,\dots,\pi_{i-1},\pi_j,\dots,\pi_n)
\end{gathered}
\end{equation*}
\\For example, $(2,5,4,3,7,1,6,0)$ in $S_8$ is transformed to $(3,7,1,2,5,4,6,0)$ using the prefix transposition $\alpha (4,7)$. The prefix transposition distance between two permutations $\pi^{\star}$ and $\pi^{\#}$ in $S_n$ is the minimum number of prefix transpositions that are needed to transform $\pi^{\star}$ into $\pi^{\#}$. An upper bound for the prefix transposition distance over all permutations in $S_n$ can be obtained by finding the upper bound to sort all permutations $\pi\in S_n$ by prefix transpositions.

Two consecutive symbols $\pi_i$ and $\pi_{i+1}$ in a permutation is said to form an \emph{adjacency} if $\pi_i+1=\pi_{i+1}\text{ (mod \emph{n})}$. By contrast, a \emph{break point} is a position $i$ in the permutation $\pi$ such that $\pi_i+1\ne \pi_{i+1}\text{ (mod \emph{n})}$. The identity permutation $I_n=(0,1,2,\dots,(n-1))$ has $n-1$ adjacencies and no breakpoints, whereas the reverse permutation $R_n=((n-1),(n-2),\dots,2,1,0)$ has $n-1$ breakpoints and no adjacencies. A transposition on a permutation can form a maximum of three adjacencies, but a prefix transposition can create only up to two adjacencies. We call a prefix transposition on a permutation as a \emph{move} in this thesis. A basic strategy to sort a permutation using prefix transpositions would be to create adjacencies in each move while reducing the number of break points. A move that does not create or destroy an adjacency is called a \emph{blank}. Moves that create one or two adjacencies each are called a \emph{single} and a \emph{double} respectively. A pictorial description of all these three types of moves is given in figure \ref{type of moves}. 

\begin{figure}[h]
\centering
\includegraphics[scale=.7]{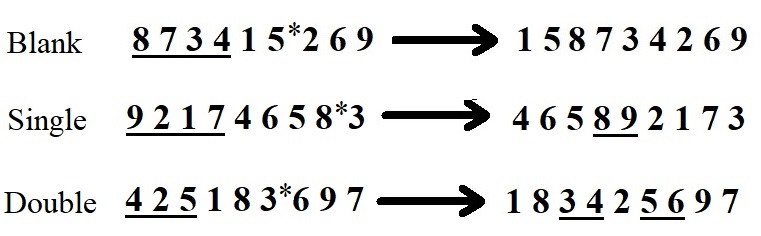}
\caption{Prefix transpositions on permutations}
\label{type of moves}
\end{figure}

For a non-identity permutation, a single is always possible, but a double may not be executable. For example, there is no prefix transposition on permutation $\pi'=(3,0,2,6,5,1,4) \in S_7$ that creates two adjacencies, whereas moving the sublist $(3,0)$ to a position before $1$ in $\pi'$ forms a single. Further, we observe that a single is not unique since moving the sublist $(3,0)$ to a position between $2$ and $6$ in $\pi'$ also forms a single, but if a double exists, it is unique.

A permutation that does not have an adjacency is said to be \emph{irreducible} or a \emph{reduced permutation}. For example, the permutation $\pi'$ defined in the previous paragraph is a reduced permutation. Any permutation $\pi^{\star}\in S_n$ that is not an irreducible can be reduced to $\pi'\in S_{n-k}$, by replacing a sublist of $k+1$ adjacent symbols in $\pi^{\star}$ by the least symbol (say $\pi_i$) in the sublist and replacing the symbols $\pi_j$ by $\pi_{j}-k$ whenever $\pi_j>\pi_i$. For example, the permutation $(4,6,1,2,3,0,5,7)\in S_8$ can be reduced to $(2,4,1,0,3,5) \in S_6$, where the sublist $(1,2,3)$ is replaced by $1$. Also, $(3,2,1,0,4) \in S_5$ can be reduced to $(3,2,1,0) \in S_4$. By Christe \cite{Christie1998}, it can be shown that sorting the permutation $\pi^{\star}\in S_n$ by prefix transpositions is equivalent to sorting the corresponding reduced permutation $\pi'\in S_{n-k}$. 

\section{Sorting Reverse Permutations}
A greedy method to sort a permutation is to create new adjacencies in each move. In 2002, Dias, Meidanis and Fortuna \cite{Dias2002} presented an algorithm to sort the reverse permutation $R_n=((n-1),(n-2),\dots,2,1,0)$ in $n-\lfloor{\frac{n}{4}\rfloor}$ prefix transpositions. They also conjectured that $n-\lfloor{\frac{n}{4}\rfloor}$ is the maximum prefix transposition distance among permutations in $S_n$. According to their algorithm, $R_8$ could be sorted in six moves. The procedure to sort $R_8$ is shown in figure \ref{r8}. Note that this algorithm did not follow a greedy approach as the first two moves of the algorithm are \emph{blanks}. The correctness of this algorithm was later proved by Vinicius Fortuna \cite{Fortuna2005} in 2005.

\begin{figure}[h]
\centering
\includegraphics[scale=.7]{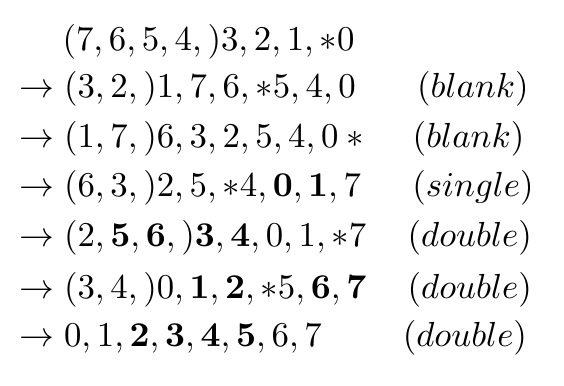}
\caption{Sorting $R_8$ in six moves}
\label{r8}
\end{figure}

Chitturi \cite{Chitturi2008} observed that the algorithm in \cite{Dias2002} to sort $R_8$ could be used for any permutation of the form $(7,\alpha,6,5,4,3,2,1,0,\beta)$, where $\alpha$ and $\beta$ are arbitrary sublists, to form seven adjacencies in six moves. This is achieved by considering the sublist $(7,\alpha)$ to be one symbol of the permutation and transforming the permutation to $(0,1,2,3,4,5,6,7,\alpha,\beta)$ in six prefix transpositions used in figure \ref{r8}. In general, we can use the algorithm by Dias et al. on a permutation of the form $R_8'=(7+i,\alpha,6+i,5+i,4+i,3+i,2+i,1+i,i,\beta)$ and make seven adjacencies in six moves.

\section{Sequence Length Algorithm}
Chitturi and Sudborough \cite{Chitturi2008} proposed the sequence length algorithm in 2008 to improve the upper bound to sort permutations using prefix transpositions. It is a greedy algorithm that performs on a reduced permutation and ensures that a single is made in each move until a double occurs. Once a double occurs, the permutation is further reduced, and the algorithm is used repeatedly until the permutation is sorted. Using this method, Chitturi et al. \cite{Chitturi2008} improved the upper bound to sort a permutation with 'n' symbols, from $n-1$ to $n-\log_8n$. All the further improvements on the upper bound to sort permutations with prefix transpositions use the sequence length algorithm.

Consider a reduced permutation $\pi=(\pi_1,\pi_2,\dots,\pi_n) \in S_n$. For any two symbols $x,y \in \pi$, the distance $dist(x,y)$ is defined as the number of positions to be traversed in the counter clockwise direction over the cyclic identity permutation $I_n$ to reach $y$ from $x$ \cite{Chitturi2012}. That is, $dist(x,y)=x-y\text{ (mod }n)$. The distance defined is clearly not symmetric as $dist(x,y)$ need not be equal to $dist(y,x)$. For example, from figure \ref{distance}, we see that for symbols in $S_8$, $dist(5,2)$ equals three whereas $dist(2,5)$ equals five.

\begin{figure}[h]
\centering
\includegraphics[scale=.7]{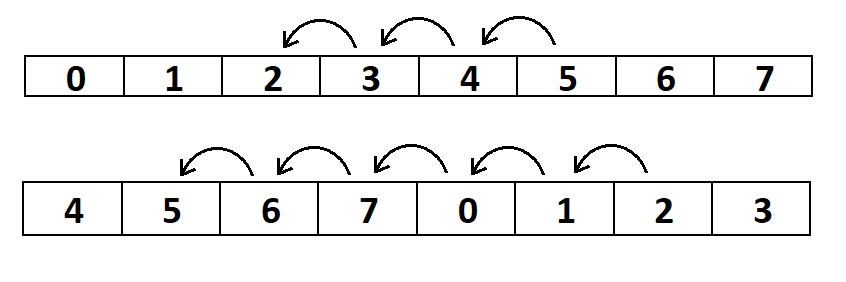}
\caption{Pictorial representation of $dist(5,2)$ and $dist(2,5)$ }
\label{distance}
\end{figure}

\subsection{Algorithm}
Consider a reduced permutation $\pi=(\pi_1,\pi_2,\dots,\pi_n) \in S_n$ and let $x$ and $y$ denote the symbols $\pi_{n-1}$ and $\pi_n$ respectively in the permutation. Note that $x \ne y-1\text{ (mod \emph{n})}$ since the permutation we consider is reduced. If $R_8'$ is a prefix of $\pi$, we use the algorithm by Dias et al. \cite{Dias2002} and stop. Otherwise, we denote the first symbol in $\pi$ as $t$ in each step and execute the following moves according to the conditions specified:\\

\emph{Case 1}: if $t=x+1\text{ (mod \emph{n})}$, the move 
\begin{equation*}
\begin{gathered}
([(x+1),\dots,(y-1)],\dots,x,*y)\rightarrow (\dots,x,(x+1),\dots,(y-1),y)
\end{gathered}
\end{equation*}
forms a double, and we stop the algorithm.

\emph{Case 2}: if $t=y+1\text{ (mod \emph{n})}$ and $x+1\text{ (mod \emph{n})}$ lies before $y-1\text{ (mod \emph{n})}$, let $s$ be the symbol in $\pi$ that immediately precedes $x+1\text{ (mod \emph{n})}$. Here we first perform the move given by
\begin{equation*}
\begin{gathered}
([(y+1),\dots,s],(x+1),\dots,(y-1),\dots,x,y*)\\
\rightarrow ((x+1),\dots,(y-1),\dots,x,y,(y+1),\dots,s)\\[3ex]
\end{gathered}
\end{equation*}
The next move follows from case 1 and forms a double in this case.

\emph{Case 3}: if $t=y+1\text{ (mod \emph{n})}$ and $x+1\text{ (mod \emph{n})}$ lies after $y-1\text{ (mod \emph{n})}$, let $s$ be the symbol in $\pi$ that immediately precedes $y-1\text{ (mod \emph{n})}$. Let $u$ be the symbol in sublist $((y+1),\dots,s)$ such that $dist(u,x+1\text{ (mod \emph{n})})=min\{dist(\pi_i,x+1\text{ (mod \emph{n})})|\pi_i \in ((y+1),\dots,s)\}$.\\
If $u \ne y+1$, let $u'$ be the symbol just before $u$ in $\pi$. Here we make a single using the move,
\begin{equation*}
\begin{gathered}
([(y+1),\dots,u'],u,\dots,s,(y-1),\dots,(x+1),\dots,x,y*)\\
\rightarrow (u,\dots,s,(y-1),\dots,(x+1),\dots,x,y,(y+1),\dots,u')\\[3ex]
\end{gathered}
\end{equation*}
If $u = y+1$, we make the move,
\begin{equation*}
\begin{gathered}
([(y+1),\dots,s],(y-1),\dots,(x+1),\dots,x,y*)\\
\rightarrow ((y-1),\dots,(x+1),\dots,x,y,(y+1),\dots,s)\\[3ex]
\end{gathered}
\end{equation*}
\emph{Case 4}: if $t \ne x+1\text{ (mod \emph{n})}$, $t \ne y+1\text{ (mod \emph{n})}$ and $t-1\text{ (mod \emph{n})}$ lies after $x+1\text{ (mod \emph{n})}$, then let $s$ be the symbol just before $x+1\text{ (mod \emph{n})}$ in $\pi$. In this case we make the move
\begin{equation*}
\begin{gathered}
([t,\dots,s],(x+1),\dots,(t-1)*,\dots,x,y)\rightarrow ((x+1),\dots,(t-1),t,\dots,s,\dots,x,y)
\end{gathered}
\end{equation*}
that creates an adjacency (single) and makes $x+1\text{ (mod \emph{n})}$ the first symbol in the transformed permutation. Note that we make this move even when the symbol $t-1=x+1\text{ (mod \emph{n})}$.

\emph{Case 5}: if $t \ne x+1\text{ (mod \emph{n})}$, $t \ne y+1\text{ (mod \emph{n})}$ and $t-1\text{ (mod \emph{n})}$ occurs before $x+1\text{ (mod \emph{n})}$, then let $s$ be the symbol just before $t-1\text{ (mod \emph{n})}$ in $\pi$. Let $u$ be the symbol in sublist $(t,\dots,s)$ such that $dist(u,x+1\text{ (mod \emph{n})})=min\{dist(\pi_i,x+1\text{ (mod \emph{n})})|\pi_i \in (t,\dots,s)\}$. Let $u'$ be the symbol just before $u$ in $\pi$. Here we make the move,
\begin{equation*}
\begin{gathered}
([t,\dots,u'],u,\dots,s,(t-1)*,\dots,(x+1),\dots,x,y)\\
\rightarrow (u,\dots,s,(t-1),t,\dots,u',\dots,(x+1),\dots,x,y)\\[3ex]
\end{gathered}
\end{equation*}
The symbols in $\pi$ that becomes the first symbol during the execution of the sequence length algorithm are called \emph{visited} symbols and they are denoted as $t_1, t_2,\dots,t_k$. The last visited symbol $t_k=x+1\text{ (mod \emph{n})}$. The symbols (except $x$ and $y$) that are not visited are called \emph{unvisited} or \emph{skipped} symbols. In 2012, Chitturi \cite{Chitturi2012} used the sequence length algorithm to relabel the symbols of $\pi \in S_n$ as 
\begin{equation*}
\begin{gathered}
(t_1,\dots,s_2,t_2,\dots,s_3,t_3,\dots,t_{i-1},\dots, s_i,t_i,\dots,x,y)\\[3ex]
\end{gathered}
\end{equation*}
We shall use this notation of $\pi$ throughout the thesis. Here $s_i$ is the skipped symbol (if it exists) that immediately precedes $t_i$. Note that the subscript $i$ of $s_i$ depends on the subscript of the next visited symbol $t_i$ and does not guarantee that $s_{i-1}$ exists. The sublist from $t_{i-1}$ to the symbol before $t_i$ in $\pi$ is called the $i^{th}$ \emph{interval} of the permutation. To illustrate the sequence length algorithm, we consider a permutation $\pi=(12, 14, 13, 8, 11, 9, 3, 6, 5, 1, 4, 7, 2, 0, 10) \in S_{15}$. Here $x=0$ and $y=10$. The visited symbols are $t_1=12, t_2=8, t_3=1$ and the regular greedy moves of the sequence length algorithm are
\begin{align*}
&([12, 14, 13], 8, 11,* 9, 3, 6, 5, 1, 4, 7, 2, 0, 10) &\\ 
\rightarrow &([8, \mathbf{11}, \mathbf{12}, 14, 13, 9, 3, 6, 5], 1, 4, 7,* 2, 0, 10) &(single)\\
\rightarrow &([1, 4, \mathbf{7}, \mathbf{8}, 11,12, 14, 13, 9], 3, 6, 5, 2, 0,* 10) &(single)\\
\rightarrow &(3, 6, 5, 2,\mathbf{ 0}, \mathbf{1}, 4, 7, 8, 11,12, 14, 13, \mathbf{9}, \mathbf{10}) &(double)
\end{align*}
Chitturi \cite{Chitturi2012} stated and proved that $\pi=(t_1,\dots,s_2,t_2,\dots,s_3,t_3,\dots,s_k,t_k,\dots,x,y)$, satisfies the following conditions:
\begin{enumerate}
\item The last visited symbol $t_k=x+1\text{ (mod \emph{n})}$.
\item $dist(t_{i+1},x)<dist(t_i,x)$, for all $i$. Further, for every skipped symbol $u$ between $t_i$ and $t_{i+1}$, $dist(u,x)>dist(t_{i+1},x)$.
\item For every $i$, the symbol $t_i-1$ lies to the right of $s_{i+1}$.
\item When permutation $\pi$ transforms into $\pi^{\star}=(t_i,\dots,s_{i+1},t_{i+1},\dots,s_k,t_k,\dots,x,y)$ by using the regular greedy moves of sequence length algorithm, the symbols to the left of $t_i$ in $\pi$ are positioned in some intervals of $\pi^{\star}$. Hence the visited and skipped symbols that are moved to a different location in the previous moves shall not become the first symbol in any of the later permutations. 
\end{enumerate}

\section{Recurrence relation and Recursive formula}
Let $\pi=(t_1,\dots,s_2,t_2,\dots,s_3,t_3,\dots,s_k,t_k,\dots,x,y)\in S_n$. Chitturi and Sudborough \cite{Chitturi2008} proved that the maximum number of moves to obtain a double in sequence length algorithm is $\frac{7(n-3)}{8}$ and the number of visited symbols in an irreducible permutation without any skipped symbol in between them is at most seven. So, if we assume that the prefix of $\pi$ is not $R_8'$ and the $i^\text {th}$ interval is the first interval to have a skipped symbol, then $i \le 8$. If we execute $R_8'$, we obtain seven adjacencies in six moves, and this is equivalent to getting a double in six moves of the sequence length algorithm. Once a double is encountered, we reduce the given permutation and continue to use the sequence length algorithm recursively until we sort the permutation.

\subsection{Recurrence relation}
Suppose that the sequence length algorithm is executed on a reduced permutation $\pi \in S_n$. Here we shall derive a recurrence relation on the number of moves required to sort a permutation with the sequence length algorithm. For $0 \le \alpha \le 1$, assume that $\alpha n$ symbols are visited and $(1-\alpha) n$ symbols are skipped, until a double is encountered. Then we have created $\alpha n+1$ adjacencies and hence resultant reduced permutation after the sequence length algorithm will contain $(n-\alpha n-1)$ symbols. The recurrence relation is given by $T(n) =T(n-\alpha n-1)+\alpha n$. Clearly this recurrence relation is bounded by the recurrence $T(n) =T(n-\alpha n)+\alpha n-1$ as a single is always possible for a permutation and this reduces the number of symbols in the reduced permutation. Solving the recurrence with $T[1]=1$, we have:\\
$T[n] =T[(1-\alpha )n]-(1-\alpha) n+n-1$\\
$T[n] =\{T[(1-\alpha )^{2}n]-(1-\alpha)^{2} n+(1-\alpha) n-1\}-(1-\alpha) n+n-1$\\
$T[n] =T[(1-\alpha )^{2}n]+n-2$\\
$T[n] =T[(1-\alpha )^{3}n]+n-3$\\
$\dots$\\
$T[n] =T[(1-\alpha )^{k}n]+n-k$\\
Now $(1-\alpha )^{k}n=1\implies n=(\frac{1}{1-\alpha})^{k} \implies k=\log_{(\frac{1}{1-\alpha})}n$\\
Hence when $k=\log_{(\frac{1}{1-\alpha})}n$, the equation $T[n] =T[(1-\alpha )^{k}n]+n-k$ becomes\\
$T[n] =T[1]+n-\log_{(\frac{1}{1-\alpha})}n$.

Thus, an upper bound to sort a permutation $\pi \in S_n$ using the sequence length algorithm is
$T(n) =n-\log_{\beta}n$ , where $\beta=(\frac{1}{1-\alpha })$

\subsection{Recursive formula}\label{Recursive formula}
This section will derive a recursive formula to find an upper bound to sort permutations by prefix transpositions. This formula uses the recurrence relation stated above. It is defined in terms of the number of symbols moved to a different position and the number of symbols skipped in the process. Thus it gives us the insight to approach better bounds. 
\begin{lemma}\label{rf}
Let $x$ be the number of visited symbols and $y$ be the number of skipped symbols that are encountered in the execution of the sequence length algorithm on a reduced permutation $\pi \in S_n$. Then $n-\log_{(\frac{x+ y}{ y})}n$ is an upper bound to sort $\pi$ with prefix transpositions
\end{lemma}
\begin{proof}
From the recurrence relation in the previous section, we know that an upper bound to sort $\pi \in S_n$ using the sequence length algorithm is $n-\log_{(\frac{1}{1-\alpha})}n$, where $0 \le \alpha \le 1$ and $\alpha n$ symbols are visited and $(1-\alpha) n$ symbols are skipped, until a double is encountered. Here we have $x=\alpha n$ and $y=(1-\alpha) n$.

\begin{equation*}
\begin{gathered}
\frac{x}{y}=\frac{\alpha n}{(1-\alpha) n}\Rightarrow \alpha =\frac{x}{x+ y}\Rightarrow \frac{1}{1-\alpha}=\frac{x+ y}{ y}\\[3ex]
\end{gathered}
\end{equation*}
Hence the base of the logarithm in the upper bound is $\frac{x+ y}{ y}$, which is the ratio of the total number of symbols moved by the number of symbols skipped until a double is encountered.
\end{proof}

\chapter{Improved Upper Bound to Sort Permutations using Prefix Transpositions}
\section{Introduction}
Dias and Meidanis \cite{Dias2002} in 2002 studied prefix transpositions. They provided the first results on the problem of sorting permutations with the minimum number of prefix transpositions. This problem was considered as a variation of the transposition distance problem. In their research paper \cite{Dias2002}, they provided lower and upper bounds of $\frac{n}{2}$ and $n-1$, respectively, to sort $\pi \in S_n$ with prefix transpositions, using a 2-approximation algorithm. In 2008, Chitturi and Sudborough \cite{Chitturi2008} used the sequence length algorithm to improve the upper bounds to $n-\log_{8}n$. The lower bound was improved to $\frac{3n}{4}$ by Labarre \cite{labarre2008edit}. The upper bounds were further improved by Chitturi to $n-\log_{(\frac{9}{2})}n$ \cite{Chitturi2012} and $n-\log_{(\frac{7}{2})}n$ \cite{Chitturi2015} in 2012 and 2015. In this chapter, we shall improve the upper bound to $n-\log_{(\frac{10}{3})}n$ \cite{nair2020improved} by introducing some alternate moves in addition to the regular greedy moves of the sequence length algorithm.

\section{Basic Notations and Overview}
Let $\pi=(t_1,\dots,s_2,t_2,\dots,s_3,t_3,\dots,s_k,t_k,\dots,x,y)$ be a permutation in $S_n$ which does not contain $R_8'$ as a prefix. The $(i-1)^\text {th}$ regular greedy move of the sequence length algorithm is given by

\begin{equation*}
\begin{gathered}
([t_{i-1},\dots,s_i],t_i,\dots,(t_{i-1}-1)*,\dots,x,y)\\
\rightarrow (t_i,\dots,(t_{i-1}-1),t_{i-1},\dots,s_i,\dots,x,y)\\[3ex]
\end{gathered}
\end{equation*}
Note that $t_{i-1}-1$ can be the next visited symbol $t_i$ or a skipped symbol after $t_i$. These moves are executed until $x+1$ becomes the first symbol. The last move is a double given by
\begin{equation*}
\begin{gathered}
([x+1,\dots,y-1],\dots,x,*y)\rightarrow (\dots,x,x+1,\dots,y-1,y)
\end{gathered}
\end{equation*}
\begin{definition}
A prefix transposition is called an alternate move if it is not a regular greedy move of the sequence length algorithm and creates at least one adjacency in the permutation.
\end{definition}
For example, in the permutation $\pi=(12, 14, 13, 8, 11, 9, 3, 6, 5, 1, 4, 7, 2, 0, 10)$, which is used to illustrate the sequence length algorithm in Section 2.3.1, an alternate move that creates a double is given by

\begin{equation*}
\begin{gathered}
([12, 14, 13, 8], 11,* 9, 3, 6, 5, 1, 4, 7, 2, 0, 10)\\ 
\rightarrow (\mathbf{11}, \mathbf{12}, 14, 13, \mathbf{8}, \mathbf{9}, 3, 6, 5, 1, 4, 7, 2, 0, 10)\\[3ex]
\end{gathered}
\end{equation*}

The basic principle used in the sequence length algorithm is to obtain a double preceded by singles in at most $\frac{7n}{8}$ moves and thus sorting the permutation faster. We define a \emph{generalised sequence length algorithm} which is a modified version of the sequence length algorithm. In this algorithm, we shall introduce some additional alternate moves to the sequence length algorithm by \cite{Chitturi2008}, thus enabling us to get a double faster. The generalised sequence length algorithm follows the same principle as the original algorithm, creating a single in each move until a double is obtained. Note that in the generalised algorithm, we may move more than one visited symbol in a move, and thus a visited symbol becomes a skipped symbol. We shall also find instances where a skipped symbol in the sequence length algorithm becomes a visited symbol in the generalised algorithm. As this algorithm is a generalisation of the sequence length algorithm, the recursive formula in Lemma \ref{rf} is applicable here too. All the lemmas and observations in the following chapters assumes that $\pi=(t_1,t_2,\dots,s_i,t_i,\dots,s_k,t_k,\dots,x,y)$ with visited and skipped symbols of the sequence length algorithm. The number of skipped symbols counted after alternate moves corresponds to the symbols skipped in the generalised sequence length algorithm. This thesis uses the terms "Sequence Length algorithm" and "Generalized Sequence Length algorithm" interchangeably, depending on the context. And the term \emph{skipped symbol} refers to a skipped symbol in either the sequence length algorithm or the generalised sequence length algorithm.

We shall assume that the $i^\text {th}$ interval is the first interval in $\pi$ to have a skipped symbol. Then by Chitturi \cite{Chitturi2008}, $i\le 8$. Furthermore, if $i\le 3$, we skip at least one symbol in two moves by the regular moves of the sequence length algorithm. This produces an upper bound of $n-\log_{3}n$ by the recursive formula for sequence length algorithm discussed in Section \ref{Recursive formula}.
\begin{observation}\label{obs3.1}
If $\pi=(t_1,t_2,\dots,s_i,t_i,\dots,s_k,t_k,\dots,x,y)$ has more than two skipped symbols in the $i^\text {th}$ interval, then an upper bound for sorting permutations by prefix transpositions is $n-\log_{(\frac{10}{3})}n$. 
\end{observation}
\begin{proof}
Suppose $\pi$ has at least three skipped symbols in the $i^\text {th}$ interval. Then in $i-1$ regular greedy moves of the sequence length algorithm, we move at least $i+2$ symbols among which at least three are skipped. Hence by the recursive formula (Section \ref{Recursive formula}), the upper bound is given by $n-\log_{(\frac{i+2}{3})}n$ which maximises the base of the logarithm to $\frac{10}{3}$ when $i=8$. Thus $n-\log_{(\frac{10}{3})}n$ is an upper bound.
\end{proof}

If the $i^\text {th}$ interval contains more than one skipped symbol, then the first skipped symbol shall be denoted by $c$. Further $t_{i-1}+1 \ne c$ and $s_i+1 \ne t_i$ since the permutation we consider is reduced. By the construction of sequence length algorithm both the skipped symbols $s_i$ and $c$ are greater than $t_{i-1}$ and so the visited symbol $t_{i-1}$ shall be denoted as $s_i-l$, for some $l\ge 1$. Thus permutation $\pi$ is denoted as $(t_1,t_2,\dots,(s_i-l),c,s_i,t_i,\dots,s_k,t_k,\dots,x,y)$.

The procedure that we adopt to improve the upper bound is to check for the symbol $s_i+1$ in $\pi=(t_1,t_2,\dots,(s_i-l),c,s_i,t_i,\dots,s_k,t_k,\dots,x,y)$ and define suitable alternate moves to get an upper bound less than or equal to $n-\log_{(\frac{10}{3})}n$. Here we consider the cases where $s_i+1$ lies to the right of $s_i$, $s_i+1$ is a skipped symbol in the $i^\text {th}$ interval ($s_i+1=c$), and $s_i+1$ lies to the left of $s_i-l$ (note that $s_i-l \ne s_i+1$, for then $s_i$ would be a visited symbol in $\pi$). Thus, combining these results with Observation \ref{obs3.1}, we claim an upper bound of $n-\log_{(\frac{10}{3})}n$ to sort permutations by prefix transpositions.

\section{Proposed Algorithm}
\begin{lemma}\label{l3.1}
Let $\pi=(t_1,t_2,\dots,(s_i-l),\dots,s_i,t_i,\dots,s_k,t_k,\dots,x,y)$ be a permutation in which $s_i+1$ lies to the right of $s_i$. Then the upper bound for sorting $\pi$ is $n-\log_{(\frac{4}{3})}n$.
\end{lemma}
\begin{proof}
Suppose that $s_i+1$ lies to the right of $s_i$. Then by sequence length algorithm $s_i+1$ is a skipped symbol and hence lies to the right of $t_i$. Then the permutation $\pi$ is in the form $(t_1,t_2,\dots,(s_i-l),\dots,s_i,t_i,\dots,(s_i+1),\dots)$. The move

\begin{equation*}
\begin{gathered}
([t_1,t_2,\dots,(s_i-l),\dots,s_i],t_i,\dots,*(s_i+1),\dots)\\
\rightarrow (t_i,\dots,t_1,t_2,\dots,(s_i-l),\dots,s_i,(s_i+1),\dots)\\[3ex]
\end{gathered}
\end{equation*}
moves at least $i$ symbols of which $i-1$ are skipped. By Lemma \ref{rf}, this produces an upper bound of $n-\log_{(\frac{i}{i-1})}n$ which maximises the base of the logarithm to $\frac{4}{3}$ when $i=4$. Thus $n-\log_{(\frac{4}{3})}n$ is an upper bound.
\end{proof}
Now we shall consider the case when $s_i+1$ lies to the left of $s_i$. Here $s_i+1$ can either be a skipped symbol in the $i^\text {th}$ interval or a visited symbol to the left of $s_i-l$. 

\begin{lemma}\label{l3.2}
Let $\pi=(t_1,t_2,\dots,(s_i-l),\dots,s_i,t_i,\dots,s_k,t_k,\dots,x,y)$ be a permutation in which $s_i+1$ is a skipped symbol in the $i^\text {th}$ interval. Then the upper bound for sorting $\pi$ is $n-\log_{(\frac{10}{3})}n$.
\end{lemma}
\begin{proof}
If the $i^\text {th}$ interval has a third skipped symbol other than $s_i$ and $s_i+1$, then by Observation \ref{obs3.1}, the upper bound is $n-\log_{(\frac{10}{3})}n$. So, we shall consider the case where $s_i$ and $s_i+1$ are the only skipped symbols in the $ i^\text {th}$ interval. Here the permutation $\pi$ equals $(t_1,t_2,\dots,(s_i-l),s_i+1,s_i,t_i,\dots,s_k,t_k,\dots,x,y)$. Now we consider the various positions of $s_i+2$ in the permutation. Clearly $s_i+2$ is either a skipped symbol that lies to the right of $t_i$ or it’s a visited symbol that lies to the left of $s_i-l$.\\

\emph{Case 1}: If $s_i+2$ lies to the right of $t_i$, the moves

\begin{equation*}
\begin{gathered}
([t_1,t_2,\dots,(s_i-l),(s_i+1)],s_i,t_i,\dots,*(s_i+2),\dots)\\
\rightarrow (s_i,t_i,\dots,t_1,t_2,\dots,(s_i-l),\dots,(s_i+1),(s_i+2),\dots)\\
([s_i],t_i,\dots,t_1,t_2,\dots,(s_i-l),\dots,*(s_i+1),(s_i+2),\dots)\\
\rightarrow (t_i,\dots,t_1,t_2,\dots,(s_i-l),\dots,s_i,(s_i+1),(s_i+2),\dots)\\[3ex]
\end{gathered}
\end{equation*}
together move $i+1$ symbols of which $i-1$ are skipped. By Lemma \ref{rf}, this produces an upper bound of $n-\log_{(\frac{i+1}{i-1})}n$ which maximises the base of logarithm to $n-\log_{(\frac{5}{3})}n$ when $i=4$.\\

\emph{Case 2}: If $s_i+2$ is a visited symbol that lies to the left of $s_i-l$. Then the permutation $\pi$ equals $(t_1,t_2,\dots,(s_i+2),\dots,(s_i-l),s_i+1,s_i,t_i,\dots,x,y)$. Consider the position of the symbol $(s_i-l+1)$ in $\pi$. If $(s_i-l+1)$ is a visited symbol, it lies just before $s_i-l$ in $\pi$. If $(s_i-l+1)$ is a skipped symbol, then it will either lie to the right of $t_i$ or in the $i^\text {th}$ interval. We shall consider each of these cases.

\emph{Case 2.1}: If $(s_i-l+1)$ lies to the right of $t_i$ then, the following 3 alternate moves

\begin{equation*}
\begin{gathered}
([t_1,\dots,(s_i+2),\dots,(s_i-l)],(s_i+1),s_i,t_i,\dots,*(s_i-l+1),\dots)\\
\rightarrow ((s_i+1),s_i,t_i,\dots,t_1,\dots,(s_i+2),\dots,(s_i-l)(s_i-l+1),\dots)\\
([(s_i+1)],s_i,t_i,\dots,*(s_i+2),\dots,(s_i-l)(s_i-l+1),\dots)\\
\rightarrow (s_i,t_i,\dots,(s_i+1),(s_i+2),\dots,(s_i-l)(s_i-l+1),\dots)\\
([s_i],t_i,\dots,*(s_i+1),(s_i+2),\dots,(s_i-l)(s_i-l+1),\dots)\\
\rightarrow (t_i,\dots,s_i,(s_i+1),(s_i+2),\dots,(s_i-l)(s_i-l+1),\dots)\\[3ex]
\end{gathered}
\end{equation*}
will move $i+1$ symbols of which $i-2$ are skipped. Thus, by Lemma \ref{rf}, the upper bound is $n-\log_{(\frac{i+1}{i-2})}n$ which maximises the base of the logarithm to $\frac{5}{2}$ when $i=4$. Thus $n-\log_{(\frac{5}{2})}n$ is an upper bound.

\emph{Case 2.2}: If $(s_i-l+1)$ lies to the left of $s_i-l$ then, we shall do the regular greedy moves of sequence length algorithm until $s_i+2$ becomes the first symbol. The next 2 moves are as follows:
\begin{equation*}
\begin{gathered}
([(s_i+2),\dots,(s_i-l+1)],(s_i-l),(s_i+1)*s_i,t_i,\dots)\\
\rightarrow((s_i-l),(s_i+1),(s_i+2),\dots,(s_i-l+1),s_i,t_i,\dots)\\
([(s_i-l),(s_i+1),\dots,s_i],t_i,\dots,(s_i-l-1)*,\dots)\\
\rightarrow (t_i,\dots,(s_i-l-1)(s_i-l),(s_i+1),\dots,s_i,\dots,)\\[3ex]
\end{gathered}
\end{equation*}
In this set of at most $i-2$ moves we move $i+1$ symbols of which at least three are skipped. Thus, by Lemma \ref{rf}, the upper bound is $n-\log_{(\frac{i+1}{3})}n$ which maximises the base of the logarithm to $3$ when $i=8$. Thus $n-\log_{3}n$ is an upper bound. Note that these moves can be executed even if $t_i=(s_i-l-1)$.

\emph{Case 2.3}: If $(s_i-l+1)$ is a skipped symbol in the $i^\text {th}$ interval, then $(s_i-l+1)=s_i \Rightarrow s_i-l=s_i-1$, since the number of skipped symbols is at most two. Then the permutation $\pi=(t_1,t_2,\dots,(s_i+2),\dots,(s_i-1),s_i+1,s_i,t_i,\dots,x,y)$. Here we do the regular greedy moves until $s_i+2$ becomes the first symbol. The next alternate move 
\begin{equation*}
\begin{gathered}
([(s_i+2),\dots,(s_i-1)](s_i+1)*s_i,t_i,\dots)\\
\rightarrow((s_i+1),(s_i+2),\dots,(s_i-1),s_i,t_i,\dots)\\[3ex]
\end{gathered}
\end{equation*}
forms a double in at most $i-2$ moves. Hence, we form eight adjacencies in at most six moves which gives an upper bound of $\frac{3n}{4}$

Hence from all the cases discussed above, the upper bound for sorting $\pi$ when $s_i+1$ is a skipped symbol in the $i^\text {th}$ interval is $n-\log_{(\frac{10}{3})}n$.
\end{proof}
\begin{lemma}\label{l3.3}
Let $\pi=(t_1,t_2,\dots,(s_i-l),\dots,s_i,t_i,\dots,s_k,t_k,\dots,x,y)$ be a permutation in which $s_i+1$ is a visited symbol that lies to the left of $s_i-l$. Then the upper bound for sorting is $n-\log_{(\frac{10}{3})}n$.
\end{lemma}
\begin{proof}
From the assumptions in the Lemma, the permutation is of the form $\pi=(t_1,t_2,\dots,(s_i+1),\dots,(s_i-l),\dots,s_i,t_i,\dots,x,y)$. Here we shall consider the position of the symbol $(s_i-l+1)$ in $\pi$.

\emph{Case 1}: If $(s_i-l+1)$ lies to the right of $t_i$ then by case 2.1 in Lemma \ref{l3.2}, the upper bound is $n-\log_{(\frac{5}{2})}n$.\\

\emph{Case 2}: If $(s_i-l+1)$ lies to the left of $s_i-l$, then $\pi=(t_1,t_2,\dots,(s_i+1),\dots,(s_i-l+1),(s_i-l),\dots,s_i,t_i,\dots,x,y)$. Here we consider the position of $(s_i-l+2)$ in $\pi$.

\emph{Case 2.1}: If $(s_i-l+2)$ lies to the right of $s_i-l$, then it is a skipped symbol. In the following two moves

\begin{equation*}
\begin{gathered}
([t_1,\dots,(s_i-l+1)],(s_i-l),\dots,s_i,t_i,\dots,*(s_i-l+2),\dots)\\
\rightarrow ((s_i-l),\dots,s_i,t_i,\dots,t_1,\dots,(s_i-l+1),(s_i-l+2),\dots)\\
([(s_i-l),\dots,s_i],t_i,\dots,(s_i-l-1)*,\dots)\\
\rightarrow(t_i,\dots,(s_i-l-1),(s_i-l),\dots,s_i,\dots)\\[3ex]
\end{gathered}
\end{equation*}
at least $i$ symbols are moved of which $i-2$ are skipped. By Lemma \ref{rf}, this produces an upper bound of $n-\log_{(\frac{i}{i-2})}n$ which maximises the base of the logarithm to $2$ when $i=4$. Thus $n-\log_{2}n$ is an upper bound. Note that these moves can be executed even if $t_i=(s_i-l-1)$.

\emph{Case 2.2}: If $(s_i-l+2)$ lies to the left of $s_i-l$, then it is a visited symbol and lies to the left of $(s_i-l+1)$ in $\pi$ (otherwise $(s_i-l+2)$ would be a skipped symbol in the $ (i-1)^\text {th}$ interval, which is a contradiction to our assumption that the first skipped symbol lies in the $i^\text {th}$ interval). Hence, the permutation takes the form $\pi=(t_1,t_2,\dots,(s_i+1),\dots,(s_i-l+2),(s_i-l+1),(s_i-l),\dots,s_i,t_i,\dots,x,y)$. Here we do the regular greedy moves until $s_i+1$ becomes the first symbol. The next 2 moves are given by 
\begin{equation*}
\begin{gathered}
([(s_i+1),\dots,(s_i-l+2),(s_i-l+1)],(s_i-l),\dots,s_i,*t_i,\dots)\\
\rightarrow ((s_i-l),\dots,s_i,(s_i+1),\dots,(s_i-l+2),(s_i-l+1),t_i,\dots)\\
([(s_i-l),\dots,(s_i-l+1)],t_i,\dots,(s_i-l-1)*,\dots)\\
\rightarrow (t_i,\dots,(s_i-l-1),(s_i-l),\dots,(s_i-l+1),\dots)\\[3ex]
\end{gathered}
\end{equation*}
These sequences of moves move at least $i$ symbols of which at least three are skipped ($s_i$ and at least two visited symbols are skipped). By Lemma \ref{rf}, the upper bound is $n-\log_{(\frac{i}{3})}n$ which maximises the base of logarithm to $n-\log_{(\frac{8}{3})}n$ when $i=8$. Note that these moves can be executed even if $t_i=(s_i-l-1)$.\\

\emph{Case 3}: If $(s_i-l+1)$ is a skipped symbol in the $i^\text {th}$ interval, then it is not the first skipped symbol in the interval, (if so, then $s_i-l$ and $(s_i-l+1)$ would form an adjacency in $\pi$). If $(s_i-l+1)$ is a skipped symbol other than $c$ and $s_i$, then we have three skipped symbols in the $i^\text {th}$ interval, and by Observation \ref{obs3.1}, the upper bound is $n-\log_{(\frac{10}{3})}n$. Now we shall suppose that $(s_i-l+1)= s_i \Rightarrow s_i-l=s_i-1$. Then $\pi=(t_1,t_2,\dots,(s_i+1),\dots,(s_i-1),c,s_i,t_i,\dots,x,y)$. Here we shall consider the position of the symbol $c+1$ in $\pi$.

\emph{Case 3.1}: If $c+1$ lies to the right of $t_1$, then in the following two moves
\begin{equation*}
\begin{gathered}
([t_1,t_2,\dots,(s_i+1),\dots,(s_i-1),c],s_i,t_i,\dots,*(c+1))\\
\rightarrow (s_i,t_i,\dots,t_1,t_2,\dots,(s_i+1),\dots,(s_i-1),c,(c+1))\\
\end{gathered}
\end{equation*}
\begin{equation*}
\begin{gathered}
([s_i],t_i,\dots,t_1,t_2,\dots,*(s_i+1),\dots,(s_i-1),c,(c+1))\\
\rightarrow (t_i,\dots,t_1,t_2,\dots,s_i,(s_i+1),\dots,(s_i-1),c,(c+1))\\[3ex]
\end{gathered}
\end{equation*}
$i+1$ symbols are moved of which $i-1$ of them are skipped. Hence by Lemma \ref{rf}, the upper bound is $n-\log_{(\frac{i+1}{i-1})}n$ which maximises the base of the logarithm to $\frac{5}{3}$ when $i=4$. Thus $n-\log_{(\frac{5}{3})}n$ is an upper bound.

\emph{Case 3.2}: If $c+1$ lies to the left of $t_i$, then it is a visited symbol and lies to the left of $s_i-l$. (Note that $c+1=s_i$ would form an adjacency in $\pi$). Further, $c$ is skipped $\Rightarrow c >s_i-1$ and so $c >s_i \Rightarrow (c +1)>(s_i+1)$. If $c +1$ lies between $s_i+1$ and $s_i-1$, it would be a skipped symbol in an interval before $i$, hence the permutation is given by $\pi=(t_1,\dots,(c+1),\dots,(s_i+1),\dots,(s_i-1), c,s_i,t_i,\dots)$. Here we do the regular greedy moves until $c+1$ becomes the first symbol. The next move 

\begin{equation*}
\begin{gathered}
([(c+1),\dots,(s_i+1),\dots,(s_i-1)],c,*s_i,t_i,\dots)\\
\rightarrow(c,(c+1),\dots,(s_i+1),\dots,(s_i-l),s_i,t_i,\dots)\\[3ex]
\end{gathered}
\end{equation*}
forms a double in at most $i-3$ moves. Hence, we form 7 adjacencies in at most 5 moves which gives an upper bound of $\frac{5n}{7}$.

Hence from the cases discussed, the upper bound for sorting permutations when $s_i+1$ is a visited symbol that lies to the left of $s_i-l$ is $n-\log_{(\frac{10}{3})}n$.
\end{proof}
In the previous lemmas, we have considered all the positions at which the symbol $s_i+1$ can exist in permutation $\pi$. Hence we have the following theorem on the upper bound to sort a permutation using prefix transpositions.
\begin{theorem}
An upper bound for sorting permutation $\pi \in S_n$ using prefix transpositions is $n-\log_{(\frac{10}{3})}n$
\end{theorem}
\begin{proof}
The result follows from Lemmas \ref{l3.1}, \ref{l3.2}, \ref{l3.3}, where we have proved that the upper bound is less than or equal to $n-\log_{(\frac{10}{3})}n$ at all possible positions for the symbol $s_i+1$ in permutation $\pi$.
\end{proof}

\chapter{Tighter Upper Bound Using Sequence Length Algorithm}
\section{Introduction}
In this chapter, we improve the upper bound to $n-\log_3n$ \cite{nair2020new} using a different set of additional alternate moves on the sequence length algorithm. A strategy similar to that in the previous chapter assumes that the first interval in $\pi$ with skipped symbols has less than four symbols. In Chapter 3, from Observation \ref{obs3.1} we saw that if the number of unvisited or skipped symbols in the $i^\text {th}$ interval is more than two, the upper bound to sort the permutation with prefix transpositions is $n-\log_{(\frac{10}{3})}n$. Here we shall first prove that if the number of skipped symbols in the $i^\text {th}$ interval is more than three, the upper bound to sort the permutation with prefix transpositions is $n-\log_{(\frac{11}{4})}n$. 
\begin{observation}\label{obs4.1}
If more than three elements are skipped in at most seven greedy moves, we move at least eleven symbols, of which at least four are skipped. This gives an upper bound of $n-\log_{(\frac{11}{4})}n< n-\log_3n$.
\end{observation}
As in the previous chapter, the following assumptions are made on permutation $\pi$ and its symbols:
\begin{enumerate}
\item The first skipped symbol in $\pi$ lies in the $i^\text {th}$ interval. Note that $i \le 8$ by the construction of sequence length algorithm.
\item $s_i$ is the last skipped symbol in $i^\text {th}$ interval.
\item $t_{i-1}=s_i-l$ for some $l\ge 1$ (if $t_{i-1}=s_i+l$ for some $l\ge 1$, then $s_i$ would be a visited symbol)
\item The maximum number of symbols in the $i^\text {th}$ interval is three. For otherwise by Observation \ref{obs4.1}, $n-\log_{(\frac{11}{4})}n$ is an upper bound.
\end{enumerate}

As in Chapter 3, here also we consider the position of the symbol $s_i+1$ in the permutation $\pi=(t_1,t_2,\dots,s_i,t_i,\dots,s_k,t_k,\dots,x,y)$ and find alternate moves to get the double faster.

\section{Proposed Algorithm}
\begin{lemma}\label{l4.1}
If $i\le 3$ then the upper bound for sorting by prefix transpositions is $n-\log_{3}n$.
\end{lemma}
\begin{proof}
Consider a permutation $\pi$ which has a skipped symbol in the second or third interval ($i\le 3$). Here we skip at least one symbol in 2 regular greedy moves. By using the same procedure as in Observation \ref{obs4.1}, we obtain an upper bound of $n-\log_{(\frac{1}{(1-\frac{2}{3})})}n=n-\log_{3}n$.
\end{proof}
\begin{lemma}\label{l4.2}
If $i\ge 4$ and $s_i+ 1$ lies to the right of $t_i$ then $n-\log_{(\frac{4}{3})}n$ is an upper bound.
\end{lemma}

\begin{proof}
Consider the move 
\begin{equation*}
\begin{gathered}
([t_1,t_2,\dots,(s_i-l),\dots,s_i],t_,\dots,*(s_i+1),\dots)
\rightarrow (t_i,\dots,t_1,\dots,s_i,(s_i+1),\dots)
\end{gathered}
\end{equation*}
that moves at least $i$ symbols of which $i-1$ are skipped. Hence by Lemma \ref{rf}, an upper bound is $n-\log_{(\frac{i}{i-1})}n$ which maximises the base of the logarithm to $\frac{4}{3}$ when $i=4$. Thus, $n-\log_{(\frac{4}{3})}n$ is an upper bound.
\end{proof}

\begin{lemma}\label{l4.3}
If $i\ge 4$ and $s_i+ 1$ is a skipped symbol that lies to the left of $t_i$ then $n-\log_{3}n$ is an upper bound.
\end{lemma}
\begin{proof}
If $s_i+ 1$ lies in the $j^\text {th}$ interval, where $j<i$, then by assumption (3), we would have considered the
symbol $s_j$ rather than $s_i$ . So, we shall assume that both $s_i+ 1$ and $s_i$ are skipped and lie in the $i^\text {th}$ interval.
Let $\pi=(t_1,t_2,\ldots,(s_i-l),\ldots,(s_i+1),s_i,t_i,\ldots,s_k,t_k,\ldots,x,y)$. Here we shall consider the position of $s_i+2$.\\

\emph{Case1}: Suppose $s_i+2$ lies to the right of $t_i$, then in the following two moves

\begin{equation*}
\begin{gathered}
([t_1,t_2,\ldots,(s_i-l),\ldots,(s_i+1)],\ldots,s_i,t_i,\ldots,*(s_i+2),\ldots)\\
\rightarrow (\ldots,s_i,t_i,\ldots,t_1,t_2,\ldots,(s_i-l),\ldots,(s_i+1),(s_i+2),\ldots)\\
([\ldots,s_i],t_i,\ldots,t_1,t_2,\ldots,(s_i-l),\ldots,*(s_i+1),(s_i+2),\ldots)\\
\rightarrow (t_i,\ldots,t_1,t_2,\ldots,(s_i-l),\ldots,s_i,(s_i+1),(s_i+2),\ldots)\\[3ex]
\end{gathered}
\end{equation*}
at least $i+1$ symbols are moved of which$i-1$ are skipped. Hence by Lemma \ref{rf}, an upper bound is
given by $n-\log_{(\frac{i+1}{i-1})}n$ which maximises the base of the logarithm to $\frac{5}{3}$ when $i=4$. Thus, an upper bound is $n-\log_{(\frac{5}{3})}n$.\\

\emph{Case 2}: if $s_i+2$ is visited symbol that lies to the left of $t_i$, then by the sequence length algorithm by Chitturi \cite{Chitturi2012}, $distance((s_i-l),x)<distance((s_i+2),x)$ and hence $s_i+2$ lies to the left of $s_i-l$. Here we consider the position of the symbol $(s_i-l+1)$ in $\pi$.

\emph{Case 2.1}: $(s_i-l+1)$ lies to the right of $t_i$ . Here we need not consider the case where $l=1$, since then $(s_i-l+1)=s_i$. Then in the following three moves
\begin{equation*}
\begin{gathered}
([t_1,t_2,\ldots,(s_i+2),\ldots,(s_i-l)],\ldots,(s_i+1),\ldots,s_i,t_i,\ldots,*(s_i-l+1),\ldots)\\
\rightarrow (\ldots,(s_i+1),\ldots,s_i,t_i,\ldots,t_1,t_2,\ldots,(s_i+2),\ldots,(s_i-l),(s_i-l+1),\ldots)\\
([\ldots,(s_i+1)],\ldots,s_i,t_i,\ldots,t_1,t_2,\ldots,*(s_i+2),\ldots,(s_i-l),(s_i-l+1),\ldots)\\
\rightarrow (\ldots,s_i,t_i,\ldots,t_1,t_2,\ldots,(s_i+1),(s_i+2),\ldots,(s_i-l),(s_i-l+1),\ldots)\\
([\ldots,s_i],t_i,\ldots,t_1,t_2,\ldots,*(s_i+1),(s_i+2),\ldots,(s_i-l),(s_i-l+1),\ldots)\\
\rightarrow (t_i,\ldots,t_1,t_2,\ldots,s_i,(s_i+1),(s_i+2),\ldots,(s_i-l),(s_i-l+1),\ldots)\\[3ex]
\end{gathered}
\end{equation*}
at least $i+1$ symbols are moved of which $i-2$ are skipped. Hence by Lemma \ref{rf}, an upper bound is
given by $n-\log_{(\frac{i+1}{i-2})}n$ which maximises the base of the logarithm to $\frac{5}{2}$ when $i=4$. Thus, $n-\log_{(\frac{5}{2})}n$ is an upper bound.

\emph{Case 2.2}: $(s_i-l+1)$ is skipped and lies to the left of $t_i$ .

(a) Suppose $l\ne1$, if $(s_i-l+1)$ is the first skipped symbol, we would have an adjacency between $s_i-l$ and $(s_i-l+1)$. To avoid the adjacency, there is at least one more skipped symbol in the interval. Hence the number of skipped symbols would be four. So $\pi=(t_1,t_2,\ldots,(s_i+2),\ldots,(s_i-l),(s_i+1),(s_i-l+1),s_i,t_i,\ldots,x,y)$. Then in the following three moves 
\begin{equation*}
\begin{gathered}
([t_1,t_2,\ldots,(s_i+2),\ldots,(s_i-l)],(s_i+1),*(s_i-l+1),s_i,t_i,\ldots)\\
\rightarrow ((s_i+1),t_1,t_2,\ldots,(s_i+2),\ldots,(s_i-l),(s_i-l+1),s_i,t_i,\ldots)\\
([(s_i+1),t_1,t_2,\ldots,(s_i+2),\ldots],(s_i-l),(s_i-l+1),s_i,*t_i,\ldots)\\
\rightarrow ((s_i-l),(s_i-l+1),s_i,(s_i+1),t_1,t_2,\ldots,(s_i+2),\ldots,t_i,\ldots)\\
([(s_i-l),(s_i-l+1),\ldots,(s_i+2),\ldots],t_i,\ldots,(s_i-l-1),*\ldots)\\
\rightarrow (t_i,\ldots,(s_i-l-1),(s_i-l),(s_i-l+1),\ldots,(s_i+2),\ldots)\\[3ex]
\end{gathered}
\end{equation*}
at least $i+2$ symbols are moved of which $i-1$ are skipped. Hence by Lemma \ref{rf}, an upper bound is given by $n-\log_{(\frac{i+2}{i-1})}n$ which maximises the base of the logarithm to $\frac{6}{3}$ when $i=4$. Thus, $n-\log_{2}n$ is an upper bound. Note that this sequence of moves can be executed even if $(s_i-l-1)=t_i$.

(b) If $l=1$, Then in the following four moves 
\begin{equation*}
\begin{gathered}
([t_1,t_2,\ldots,(s_i+2),\ldots,(s_i-1)],c,(s_i+1),*s_i,t_i,\ldots)\\
\rightarrow (c,(s_i+1),t_1,t_2,\ldots,(s_i+2),\ldots,(s_i-1),s_i,t_i,\ldots)\\
([c],(s_i+1),t_1,t_2,\ldots,(s_i+2),\ldots,(s_i-1),s_i,t_i,\ldots,*(c+1),\ldots)\\
\rightarrow ((s_i+1),t_1,t_2,\ldots,(s_i+2),\ldots,(s_i-1),s_i,t_i,\ldots,c,(c+1),\ldots)\\[3ex]
\end{gathered}
\end{equation*}
(we make the above move irrespective of the position of $c+1$, even though here $c+1$ is assumed to lie to the right of $t_i$).

\begin{equation*}
\begin{gathered}
([(s_i+1),t_1,t_2,\ldots,(s_i+2),\ldots],(s_i-1),s_i,*t_i,\ldots)\\
\rightarrow ((s_i-1),s_i,(s_i+1),t_1,t_2,\ldots,(s_i+2),\ldots,t_i,\ldots)\\
([(s_i-1),s_i,(s_i+1),t_1,t_2,\ldots,(s_i+2),\ldots],t_i,\ldots,(s_i-2),*\ldots)\\
\rightarrow (t_i,\ldots,(s_i-2),(s_i-1),s_i,(s_i+1),t_1,t_2,\ldots,(s_i+2),\ldots,)\\[3ex]
\end{gathered}
\end{equation*}
at least $i+2$ symbols are moved of which $i-2$ are skipped. Hence by Lemma \ref{rf}, an upper bound is given by $n-\log_{(\frac{i+2}{i-2})}n$ which maximises the base of the logarithm to $\frac{6}{2}$ when $i=4$. Thus, $n-\log_{3}n$ is an upper bound. Note that if $c$ does not exist or lies between $s_i+1$ and $s_i$, then the second move can be omitted, and hence we obtain a better upper bound. Further this sequence of moves can be executed even if $s_i-2=t_i$.

\emph{Case 2.3}: $(s_i-l+1)$ is visited and lies to the left of $t_i$. Further $(s_i-l+1)$ lies immediately to the left of $s_i-l$, otherwise there is another visited symbol between them, then $s_i-l$ would be skipped by the construction of sequence length
algorithm. Here $\pi=(t_1,t_2,\ldots,(s_i+2),\ldots,(s_i-l+1),(s_i-l),\ldots,(s_i+1),\ldots,s_i,t_i,\ldots,x,y)$. Then we follow the usual greedy moves for all visited symbols until $s_i+2$ becomes the first symbol. Then we do the following moves 
\begin{equation*}
\begin{gathered}
[(s_i+2),\ldots,(s_i-l+1)],(s_i-l),\ldots,(s_i+1)*, \ldots,s_i,t_i,\ldots)\\
\rightarrow (s_i-l),\ldots,(s_i+1), (s_i+2),\ldots,(s_i-l+1), \ldots,s_i,t_i,\ldots)\\
[(s_i-l),\ldots,(s_i+1), (s_i+2),\ldots,(s_i-l+1), \ldots,s_i],t_i,\ldots,(s_i-l-1)*,\ldots)\\
\rightarrow t_i,\ldots,(s_i-l-1),(s_i-l),\ldots,(s_i+1), (s_i+2),\ldots,(s_i-l+1), \ldots,s_i,\ldots)\\[3ex]
\end{gathered}
\end{equation*}
Here in at most $i-2$ moves at least $i+1$ symbols are moved of which three are skipped. Hence by Lemma \ref{rf}, an upper bound is given by $n-\log_{(\frac{i+1}{3})}n$ which maximises the base of the logarithm to $\frac{9}{3}$ when $i=8$. Thus, $n-\log_{3}n$ is an upper bound.\\

\emph{Case 3}: if $s_i+2$ is skipped and lies to the left of $t_i$, we consider the position of the symbol $s_i+3$ in $\pi=(t_1,t_2,\ldots,(s_i-l),(s_i+2),(s_i+1),s_i,t_i,\ldots,x,y)$. Note that if $s_i+3$ is a skipped symbol in the $i^\text {th}$ interval, then the $i^\text {th}$ interval has four skipped symbols and by Observation \ref{obs4.1}, $n-\log_{(\frac{11}{4})}n$ is an upper bound.

\emph{Case 3.1}: $s_i+3$ lies to the right of $t_i$. Then in the following three moves
\begin{equation*}
\begin{gathered}
([t_1,t_2,\ldots,(s_i-l),(s_i+2)],(s_i+1),s_i,t_i,\ldots,*(s_i+3),\ldots)\\
\rightarrow ((s_i+1),s_i,t_i,\ldots,t_1,t_2,\ldots,(s_i-l),(s_i+2),(s_i+3),\ldots)\\
([(s_i+1)],s_i,t_i,\ldots,t_1,t_2,\ldots,(s_i-l),*(s_i+2),(s_i+3),\ldots)\\
\rightarrow (s_i,t_i,\ldots,t_1,t_2,\ldots,(s_i-l),(s_i+1),(s_i+2),(s_i+3),\ldots)\\
([s_i],t_i,\ldots,t_1,t_2,\ldots,(s_i-l),*(s_i+1),(s_i+2),(s_i+3),\ldots)\\
\rightarrow (t_i,\ldots,t_1,t_2,\ldots,(s_i-l),s_i,(s_i+1),(s_i+2),(s_i+3),\ldots)\\[3ex]
\end{gathered}
\end{equation*}
at least $i+2$ symbols are moved of which $i-1$ are skipped. Hence by Lemma \ref{rf}, an upper bound is given by $n-\log_{(\frac{i+2}{i-1})}n$ which maximises the base of the logarithm to $\frac{6}{3}$ when $i=4$. Thus, $n-\log_{2}n$ is an upper bound.

\emph{Case 3.2}: $s_i+3$ is visited and lies to the left of $t_i$. Then we consider the position of the symbol $(s_i-l+1)$ in $\pi$.

(a) $(s_i-l+1)$ lies to the right of $t_i$. Here we need not consider the case where $l=1$, since then $(s_i-l+1)=s_i$. Then in the following four moves
\begin{equation*}
\begin{gathered}
([t_1,t_2,\ldots,(s_i+3),\ldots,(s_i-l)],(s_i+2),(s_i+1),s_i,t_i,\ldots,*(s_i-l+1),\ldots)\\
\rightarrow ((s_i+2),(s_i+1),s_i,t_i,\ldots,t_1,t_2,\ldots,(s_i+3),\ldots,(s_i-l),(s_i-l+1),\ldots)\\
([(s_i+2)],(s_i+1),s_i,t_i,\ldots,t_1,t_2,\ldots,*(s_i+3),\ldots,(s_i-l),(s_i-l+1),\ldots)\\
\rightarrow ((s_i+1),s_i,t_i,\ldots,t_1,t_2,\ldots,(s_i+2),(s_i+3),\ldots,(s_i-l),(s_i-l+1),\ldots)\\
\end{gathered}
\end{equation*}
\begin{equation*}
\begin{gathered}
([(s_i+1)],s_i,t_i,\ldots,t_1,t_2,\ldots,*(s_i+2),(s_i+3),\ldots,(s_i-l),(s_i-l+1),\ldots)\\
\rightarrow (s_i,t_i,\ldots,t_1,t_2,\ldots,(s_i+1),(s_i+2),(s_i+3),\ldots,(s_i-l),(s_i-l+1),\ldots)\\
([s_i],t_i,\ldots,t_1,t_2,\ldots,*(s_i+1),(s_i+2),(s_i+3),\ldots,(s_i-l),(s_i-l+1),\ldots)\\
\rightarrow (t_i,\ldots,t_1,t_2,\ldots,s_i,(s_i+1),(s_i+2),(s_i+3),\ldots,(s_i-l),(s_i-l+1),\ldots)\\[3ex]
\end{gathered}
\end{equation*}
at least $i+2$ symbols are moved of which $i-2$ are skipped. Hence by Lemma \ref{rf}, an upper bound is given by $n-\log_{(\frac{i+2}{i-2})}n$ which maximises the base of the logarithm to $\frac{6}{2}$ when $i=4$. Thus, $n-\log_{3}n$ is an upper bound.

(b) $(s_i-l+1)$ lies to the right of $t_i$. Then we follow the usual greedy moves for all visited symbols until $s_i+3$ becomes the first symbol. Then we execute the following couple of moves
\begin{equation*}
\begin{gathered}
([(s_i+3),\ldots,(s_i-l+1)],(s_i-l),(s_i+2),*(s_i+1),s_i,t_i,\ldots)\\
\rightarrow ((s_i-l),(s_i+2),(s_i+3),\ldots,(s_i-l+1),(s_i+1),s_i,t_i,\ldots)\\
([(s_i-l),(s_i+2),(s_i+3),\ldots,(s_i-l+1),(s_i+1),s_i],t_i,\ldots,(s_i-l-1)*,\ldots)\\
\rightarrow (t_i,\ldots,(s_i-l-1),(s_i-l),(s_i+2),(s_i+3),\ldots,(s_i-l+1),(s_i+1),s_i,\ldots)\\[3ex]
\end{gathered}
\end{equation*}
Here in at most $i-2$ moves at least $i+2$ symbols are moved of which four are skipped. Hence by Lemma \ref{rf}, an upper bound is given by $n-\log_{(\frac{i+2}{4})}n$ which maximises the base of the logarithm to $\frac{5}{2}$ when $i=8$. Thus, $n-\log_{(\frac{5}{2})}n$ is an upper bound.

(c) $(s_i-l+1)$ is skipped and lies to the left of $t_i$. Here we explore only the case when $l=1$ otherwise $l\ne1$ the number of skipped symbols exceed three. Then by Observation \ref{obs4.1}, $n-\log_{(\frac{11}{4})}n$ is an upper bound. Now $\pi=(t_1,t_2,\ldots,(s_i+3),(s_i-1),(s_i+2),(s_i+1),s_i,t_i,\ldots,x,y)$. Here we shall consider two subcases:

(i) When $s_i+4$ lies to the right of $t_i$, in the following two moves, $i+2$ symbols are moved of which $i$ are skipped.
\begin{equation*}
\begin{gathered}
([t_1,t_2,\ldots,(s_i+3)],(s_i-1),(s_i+2),(s_i+1),s_i,t_i,\ldots,*(s_i+4),\ldots)\\
\rightarrow ((s_i-1),(s_i+2),(s_i+1),s_i,t_i,\ldots,t_1,t_2,\ldots,(s_i+3),(s_i+4),\ldots)\\
([(s_i-1),(s_i+2),(s_i+1),s_i],t_i,\ldots,t_1,t_2,\ldots,(s_i+3),(s_i+4),\ldots,(s_i-2)*,\ldots)\\
\rightarrow (t_i,\ldots,t_1,t_2,\ldots,(s_i+3),(s_i+4),\ldots,(s_i-2),(s_i-1),(s_i+2),(s_i+1),s_i,\ldots)\\[3ex]
\end{gathered}
\end{equation*}
Hence by Lemma \ref{rf}, an upper bound is given by $n-\log_{(\frac{i+2}{i})}n$ which maximises the base of the logarithm to $\frac{6}{4}$ when $i=8$. Thus, $n-\log_{(\frac{3}{2})}n$ is an upper bound. Further, note that this sequence of moves can be executed even if $s_i-2=t_i$.

(ii) When $s_i+4$ is visited and lies to the left of $t_i$, the permutation takes the form $\pi=(t_1,t_2,\ldots,(s_i+4),(s_i+3),(s_i-1),(s_i+2),(s_i+1),s_i,t_i,\ldots,x,y)$. Then in the following four moves
\begin{equation*}
\begin{gathered}
([t_1,t_2,\ldots,(s_i+4),(s_i+3),(s_i-1)],(s_i+2),(s_i+1),*s_i,t_i,\ldots)\\
\rightarrow ((s_i+2),(s_i+1),t_1,t_2,\ldots,(s_i+4),(s_i+3),(s_i-1),s_i,t_i,\ldots)\\
\end{gathered}
\end{equation*}
\begin{equation*}
\begin{gathered}
([(s_i+2)],(s_i+1),t_1,t_2,\ldots,(s_i+4),*(s_i+3),(s_i-1),s_i,t_i,\ldots)\\
\rightarrow ((s_i+1),t_1,t_2,\ldots,(s_i+4),(s_i+2),(s_i+3),(s_i-1),s_i,t_i,\ldots)\\
([(s_i+1),t_1,t_2,\ldots,(s_i+4),(s_i+2),(s_i+3)],(s_i-1),s_i,*t_i,\ldots)\\
\rightarrow ((s_i-1),s_i,(s_i+1),t_1,t_2,\ldots,(s_i+4),(s_i+2),(s_i+3),t_i,\ldots)\\
([(s_i-1),s_i,(s_i+1),t_1,t_2,\ldots,(s_i+4),(s_i+2),(s_i+3)],t_i,\ldots,(s_i-2)*,\ldots)\\
\rightarrow (t_i,\ldots,(s_i-2),(s_i-1),s_i,(s_i+1),t_1,t_2,\ldots,(s_i+4),(s_i+2),(s_i+3),\ldots)\\[3ex]
\end{gathered}
\end{equation*}
at least $i+2$ symbols are moved of which $i-2$ are skipped. Hence by Lemma \ref{rf}, an upper bound is given by $n-\log_{(\frac{i+2}{i-2})}n$ which maximises the base of the logarithm to $\frac{6}{2}$ when $i=4$. Thus, $n-\log_{3}n$ is an upper bound.

Hence from each of the cases in the lemma, an upper bound for $i\ge 4$ when $s_i+ 1$ is a skipped symbol that lies to the left of $t_i$ is $n-\log_{3}n$
\end{proof}
\begin{lemma}\label{l4.4}
If $i\ge 4$ and $s_i+ 1$ is a visited symbol that lies to the left of $t_i$ then $n-\log_{3}n$ is an upper bound.
\end{lemma}
\begin{proof}
$\pi=(t_1,t_2,\ldots,(s_i+1),\ldots,(s_i-l),\ldots,s_i,t_i,\ldots,x,y)$. Here we shall consider the position of the symbol $(s_i-l+1)$ in $\pi$.\\

\emph{Case 1}: $(s_i-l+1)$ lies to the right of $t_i$. Here we need not consider the case where $l=1$, since then $(s_i-l+1)=s_i$. Then in the following two moves
\begin{equation*}
\begin{gathered}
([t_1,t_2,\ldots,(s_i+1),\ldots,(s_i-l)],\ldots,s_i,t_i,\ldots,*(s_i-l+1),\ldots)\\
\rightarrow (\ldots,s_i,t_i,\ldots,t_1,t_2,\ldots,(s_i+1),\ldots,(s_i-l),(s_i-l+1),\ldots)\\
\end{gathered}
\end{equation*}
\begin{equation*}
\begin{gathered}
([\ldots,s_i],t_i,\ldots,t_1,t_2,\ldots,*(s_i+1),\ldots,(s_i-l),(s_i-l+1),\ldots)\\
\rightarrow (t_i,\ldots,t_1,t_2,\ldots,s_i,(s_i+1),\ldots,(s_i-l),(s_i-l+1),\ldots)\\[3ex]
\end{gathered}
\end{equation*}
at least $i$ symbols are moved of which $i-2$ are skipped. Hence by Lemma \ref{rf}, an upper bound is given by $n-\log_{(\frac{i}{i-2})}n$ which maximises the base of the logarithm to $\frac{4}{2}$ when $i=4$. Thus, $n-\log_{2}n$ is an upper bound.\\

\emph{Case 2}: $(s_i-l+1)$ is visited and lies to the left of $t_i$. Here we need not consider the case where $l=1$, since then $(s_i-l+1)=s_i$. Consider the position of the symbol $(s_i-l+2)$ in $\pi$.

\emph{Case 2.1}: $(s_i-l+2)$ lies to the right of $(s_i-l+1)$. Here $(s_i-l+2)$ may either be a skipped symbol in the $i^\text {th}$ interval or lie to the right of the visited symbol $t_i$. Then in the following two moves
\begin{equation*}
\begin{gathered}
([t_1,t_2,\ldots,(s_i+1),\ldots,(s_i-l+1)],(s_i-l),\ldots,s_i,t_i,\ldots,*(s_i-l+2),\ldots)\\
\rightarrow ((s_i-l),\ldots,s_i,t_i,\ldots,t_1,t_2,\ldots,(s_i+1),\ldots,(s_i-l+1),(s_i-l+2),\ldots)\\
([(s_i-l),\ldots,s_i],t_i,\ldots,(s_i-l-1)*,\ldots)
\rightarrow (t_i,\ldots,(s_i-l-1),(s_i-l),\ldots,s_i,\ldots)\\[3ex]
\end{gathered}
\end{equation*}
at least $i$ symbols are moved of which $i-2$ are skipped. Hence by Lemma \ref{rf}, an upper bound is given by $n-\log_{(\frac{i}{i-2})}n$ which maximises the base of the logarithm to $\frac{4}{2}$ when $i=4$. Thus, $n-\log_{2}n$ is an upper bound. Further note that this sequence of moves can be executed even if $(s_i-l-1)=t_i$.

\emph{Case 2.2}: $(s_i-l+2)$ is a visited symbol and lies to the left of $(s_i-l+1)$. Follow the usual greedy moves for all visited symbols until $s_i+1$ becomes the first symbol. Then we make the following move.

\begin{equation*}
\begin{gathered}
([(s_i+1),\ldots,(s_i-l+2),(s_i-l+1)],(s_i-l),\ldots,s_i,*t_i,\ldots)\\
\rightarrow ((s_i-l),\ldots,s_i,(s_i+1),\ldots,(s_i-l+2),(s_i-l+1),t_i,\ldots)\\[3ex]
\end{gathered}
\end{equation*}
where at least 2 visited symbols are skipped. So here in at most $i-3$ at least $i$ symbols are moved of which three are skipped. Hence by Lemma \ref{rf}, an upper bound is given by $n-\log_{(\frac{i}{3})}n$ which maximises the base of the logarithm to $\frac{8}{3}$ when $i=8$. Thus, $n-\log_{(\frac{8}{3})}n$ is an upper bound.\\

\emph{Case 3}: $(s_i-l+1)$ is skipped and lies to the left of $t_i$. Then $(s_i-l+1)$ cannot be the first skipped element, for then $s_i-l$ and $(s_i-l+1)$ would be consecutive symbols in the permutation and hence form an adjacency. In this case we shall find an upper bound for two different subcases - (1) $(s_i-l+1)\ne s_i$ (when $l \ne 1$) and (2) $(s_i-l+1)= s_i$ (when $l = 1$).

Consider $(s_i-l+1)\ne s_i$. Then the permutation $\pi$ equals $(t_1,t_2,\ldots,(s_i+1),\ldots,(s_i-l),c,(s_i-l+1),s_i,t_i,\ldots,s_k,t_k,\ldots,x,y)$. Consider the symbol $(s_i-l+2)$ in $\pi$.

\emph{Case 3.1}: $(s_i-l+2)$ lies to the right of $t_i$. Consider the following two moves
\begin{equation*}
\begin{gathered}
([t_1,t_2,\ldots,(s_i+1),\ldots,(s_i-l),c,(s_i-l+1)],s_i,t_i,\ldots,*(s_i-l+2),\ldots)\\
\rightarrow (s_i,t_i,\ldots,t_1,t_2,\ldots,(s_i+1),\ldots,(s_i-l),c,(s_i-l+1),(s_i-l+2),\ldots)\\
([s_i],t_i,\ldots,t_1,t_2,\ldots,*(s_i+1),\ldots,(s_i-l),c,(s_i-l+1),(s_i-l+2),\ldots)\\
\rightarrow (t_i,\ldots,t_1,t_2,\ldots,s_i,(s_i+1),\ldots,(s_i-l),c,(s_i-l+1),(s_i-l+2),\ldots)\\
\end{gathered}
\end{equation*}
in which $i+2$ symbols are moved of which $i$ are skipped. Hence by Lemma \ref{rf}, an upper bound is given by $n-\log_{(\frac{i+2}{i})}n$ which maximises the base of the logarithm to $\frac{6}{4}$ when $i=4$. Thus, $n-\log_{(\frac{3}{2})}n$ is an upper bound.

\emph{Case 3.2}: $(s_i-l+2)$ is a visited symbol and lies to the left of $t_i$. Follow the usual greedy moves for all visited symbols
until $s_i+1$ becomes the first symbol. Then we do the following move where at least one visited symbol is skipped.
\begin{equation*}
\begin{gathered}
([(s_i+1),\ldots,(s_i-l+2)](s_i-l),c,(s_i-l+1)],s_i,*t_i,\ldots)\\
\rightarrow ((s_i-l),c,(s_i-l+1)],s_i,(s_i+1),\ldots,(s_i-l+2),t_i,\ldots)\\
([(s_i-l),c,(s_i-l+1)],s_i,(s_i+1),\ldots,(s_i-l+2)]t_i,\ldots,(s_i-l-1)*,\ldots)\\
\rightarrow (t_i,\ldots,(s_i-l-1),(s_i-l),c,(s_i-l+1)],s_i,(s_i+1),\ldots,(s_i-l+2),\ldots)\\[3ex]
\end{gathered}
\end{equation*}
Note that we can execute these moves even when $(s_i-l-1)=t_i$. So here in at most $i-2$ moves at least $i+2$ symbols are moved of which four are skipped. Hence by Lemma \ref{rf}, an upper bound is given by $n-\log_{(\frac{i+2}{4})}n$ which maximises the base of the logarithm to $\frac{10}{4}$ when $i=8$. Thus, $n-\log_{(\frac{5}{2})}n$ is an upper bound.

\emph{Case 3.3}: $(s_i-l+2)$ is skipped and lies to the left of $t_i$. Then permutation $\pi$ equals $(t_1,t_2,\ldots,(s_i+1),\ldots,(s_i-l),(s_i-l+2),(s_i-l+1),s_i,t_i,\ldots,s_k,t_k,\ldots,x,y)$. The symbol $(s_i-l+3)$ lies either to the right of $t_i$ or shall be a visited symbol after $s_i+1$. Note: if $(s_i-l+3)=s_i+1$ then $(s_i-l+2)=s_i$ and hence $(s_i-l+1)(s_i-l+2)$ would form an adjacency.

(a) If $(s_i-l+3)$ lies to the right of $t_i$, then in the following three moves
\begin{equation*}
\begin{gathered}
([t_1,t_2,\ldots,(s_i+1),\ldots,(s_i-l),(s_i-l+2)],(s_i-l+1),s_i,t_i,\ldots,*(s_i-l+3),\ldots)\\
\rightarrow ((s_i-l+1),s_i,t_i,\ldots,t_1,t_2,\ldots,(s_i+1),\ldots,(s_i-l),(s_i-l+2),(s_i-l+3),\ldots)\\
([(s_i-l+1)],s_i,t_i,\ldots,t_1,t_2,\ldots,(s_i+1),\ldots,(s_i-l),*(s_i-l+2),(s_i-l+3),\ldots)\\
\rightarrow (s_i,t_i,\ldots,t_1,t_2,\ldots,(s_i+1),\ldots,(s_i-l),(s_i-l+1),(s_i-l+2),(s_i-l+3),\ldots)\\
([s_i],t_i,\ldots,t_1,t_2,\ldots,*(s_i+1),\ldots,(s_i-l),(s_i-l+1),(s_i-l+2),(s_i-l+3),\ldots)\\
\rightarrow (t_i,\ldots,t_1,t_2,\ldots,s_i,(s_i+1),\ldots,(s_i-l),(s_i-l+1),(s_i-l+2),(s_i-l+3),\ldots)\\[3ex]
\end{gathered}
\end{equation*}
$i+2$ symbols are moved of which $i-1$ are skipped. Hence by Lemma \ref{rf}, an upper bound is given by $n-\log_{(\frac{i+2}{i-1})}n$ which maximises the base of the logarithm to $\frac{6}{3}$ when $i=4$. Thus, $n-\log_{2}n$ is an upper bound.

(b) If $(s_i-l+3)$ is a visited symbol after $s_i+1$. Then $\pi=(t_1,t_2,\ldots,(s_i+1),\ldots,(s_i-l+3),(s_i-l),(s_i-l+2),(s_i-l+1),s_i,t_i,\ldots,s_k,t_k,\ldots,x,y)$. Here we follow the usual greedy moves for all visited symbols until $s_i+1$ becomes the first symbol. Then we do the following move where at least one visited symbol is skipped.
\begin{equation*}
\begin{gathered}
([(s_i+1),\ldots,(s_i-l+3)](s_i-l),(s_i-l+2),(s_i-l+1)],s_i,*t_i,\ldots)\\
\rightarrow ((s_i-l),(s_i-l+2),(s_i-l+1),s_i,(s_i+1),\ldots,(s_i-l+3),t_i,\ldots)\\
([(s_i-l),(s_i-l+2),\ldots,(s_i-l+3)],t_i,\ldots,(s_i-l-1)*,\ldots)\\
\rightarrow (t_i,\ldots,(s_i-l-1),(s_i-l),(s_i-l+2),\ldots,(s_i-l+3),\ldots)\\[3ex]
\end{gathered}
\end{equation*}
Note that we can execute these moves even when $(s_i-l-1)=t_i$. So here in at most $i-2$ moves at least $i+2$ symbols are moved of which four are skipped. Hence by Lemma \ref{rf}, an upper bound is given by $n-\log_{(\frac{i+2}{4})}n$ which maximises the base of the logarithm to $\frac{10}{4}$ when $i=8$. Thus, $n-\log_{(\frac{5}{2})}n$ is an upper bound.

Now we shall consider the second part of the proof where $(s_i-l+1)= s_i$ (when $l = 1$). Clearly $s_i$ is not the only skipped symbol in the $i^\text{th}$ interval because then there would be an adjacency between $s_i-1$ and $s_i$ . Let $c$ be the first skipped symbol in the $i^\text{th}$ interval. Then $\pi=(t_1,t_2,\ldots,(s_i+1),\ldots,(s_i-1),c,\ldots,s_i,t_i,\ldots,s_k,t_k,\ldots,x,y)$.

(a) Suppose that $c+1$ lies to the right of $t_i$. Consider the following two moves:

\begin{equation*}
\begin{gathered}
([t_1,t_2,\ldots,(s_i+1),\ldots,(s_i-1),c],\ldots,s_i,t_i,\ldots,*(c+1),\ldots)\\
\rightarrow (\ldots,s_i,t_i,\ldots,t_1,t_2,\ldots,(s_i+1),\ldots,(s_i-1),c,(c+1),\ldots)\\
([\ldots,s_i],t_i,\ldots,t_1,t_2,\ldots,*(s_i+1),\ldots,(s_i-1),c,(c+1),\ldots)\\
\rightarrow (t_i,\ldots,t_1,t_2,\ldots,s_i,(s_i+1),\ldots,(s_i-1),c,(c+1),\ldots)\\[3ex]
\end{gathered}
\end{equation*}
They move $i+1$ symbols of which $i-1$ are skipped. Hence by Lemma \ref{rf}, an upper bound is given by $n-\log_{(\frac{i+1}{i-1})}n$ which maximises the base of the logarithm to $\frac{5}{3}$ when $i=4$. Thus, $n-\log_{(\frac{5}{3})}n$ is an upper bound.

(b) Suppose that $c-1$ lies to the right of $t_i$. Consider the following two moves:

\begin{equation*}
\begin{gathered}
([t_1,t_2,\ldots,(s_i+1),\ldots,(s_i-1)],c,\ldots,*s_i,t_i,\ldots,(c-1),\ldots)\\
\rightarrow (c,\ldots,t_1,t_2,\ldots,(s_i+1),\ldots,(s_i-1),s_i,t_i,\ldots,(c-1),\ldots)\\
([c,\ldots,t_1,t_2,\ldots,(s_i+1),\ldots,(s_i-1),s_i],t_i,\ldots,(c-1)*,\ldots)\\
\rightarrow (t_i,\ldots,(c-1),c,\ldots,t_1,t_2,\ldots,(s_i+1),\ldots,(s_i-1),s_i,\ldots)\\[3ex]
\end{gathered}
\end{equation*}
They move $i+1$ symbols of which $i-1$ are skipped. Hence by Lemma \ref{rf}, an upper bound is given by $n-\log_{(\frac{i+1}{i-1})}n$ which maximises the base of the logarithm to $\frac{5}{3}$ when $i=4$. Thus, $n-\log_{(\frac{5}{3})}n$ is an upper bound.

(c) $c+1$ and $c-1$ lie to the left of $t_i$. If $c+1$ is a skipped element, then there should be at least one more skipped symbol between $c$ and $c+1$, else we get an adjacency between $c$ and $c+1$ in $\pi$. So, then the number of skipped elements becomes four. Hence, we shall assume that $c+1$ is a visited symbol. Now we shall consider two subcases according to the position of $c-1$ in $\pi$. 

(i) When $c-1$ is a visited symbol. Then the permutation $\pi$ equals $(t_1,t_2,\ldots,(c+1),(c-1),\ldots,(s_i-1),c,\ldots,s_i,t_i,\ldots,s_k,t_k,\ldots,x,y)$. The first move that we execute here is 
\begin{equation*}
\begin{gathered}
([t_1,t_2,\ldots,(c+1),(c-1)],t_j,\ldots,(s_i-1),*c,\ldots,s_i,t_i,\ldots)\\
\rightarrow (t_j,\ldots,(s_i-1),t_1,t_2,\ldots,(c+1),(c-1),c,\ldots,s_i,t_i,\ldots)\\[3ex]
\end{gathered}
\end{equation*}
This move skips at least one visited symbol (namely $c-1$). Later we execute the regular greedy moves from the visited symbol $t_j$. In this sequence of at most $i-2$ moves, we move at least $i+1$ symbols and skip three. Hence by Lemma \ref{rf}, an upper bound is given by $n-\log_{(\frac{i+1}{3})}n$ which maximises the base of the logarithm to $\frac{9}{3}$ when $i=8$. Thus, $n-\log_{3}n$ is an upper bound.

(ii) When $c-1$ is skipped, $\pi$ becomes $(t_1,t_2,\ldots,(c+1),\ldots,(s_i+1),\ldots,(s_i-1),c,(c-1),s_i,t_i,\ldots,s_k,t_k,\ldots,x,y)$. Here we follow the regular greedy moves until $c+1$ becomes the first symbol. The next move is 
\begin{equation*}
\begin{gathered}
([(c+1),\ldots,(s_i+1)],t_j,\ldots,(s_i-1),c,*(c-1),s_i,t_i,\ldots)\\
\rightarrow (t_j,\ldots,(s_i-1),c,(c+1),\ldots,(s_i+1),(c-1),s_i,t_i,\ldots)\\[3ex]
\end{gathered}
\end{equation*}
which skips at least one visited symbol (namely $s_i+1$). Later we execute the regular greedy moves from the visited symbol $t_j$. In this sequence of at most $i-2$ moves, we move at least $i+2$ symbols and skip four. Hence by Lemma \ref{rf}, an upper bound is given by $n-\log_{(\frac{i+2}{4})}n$ which maximises the base of the logarithm to $\frac{10}{4}$ when $i=8$. Thus, $n-\log_{(\frac{5}{2})}n$ is an upper bound.

By considering all the cases discussed in the lemma, an upper bound for $i\ge4$ where $s_i+1$ is a visited symbol and lies to the left of $t_i$ is $n-\log_{3}n$.
\end{proof}

\begin{theorem}
$n-\log_{3}n$ is an upper bound to sort permutations with prefix transpositions.
\end{theorem}
\begin{proof}
If $R'_8$ is a prefix of $\pi$ then the corresponding upper bound is $\frac{3n}{4}$. If $R'_8$ is not a prefix, we encounter at least one skipped symbol $s_i$, where $1\le i\le8$. If more than 3 elements are skipped in at most 7 greedy moves, then by Observation \ref{obs4.1} we obtain an upper bound of $n-\log_{(\frac{11}{4})}n$. If the number of skipped symbols is at most three, Lemmas \ref{l4.1}, \ref{l4.2}, \ref{l4.3} and \ref{l4.4} prove an upper bound of $n-\log_{3}n$. Hence the theorem.
\end{proof}

\chapter{$\mathbf{\emph{n}-log_2\emph{n}}$ Upper Bound to Sort Permutations with Prefix Transpositions using Blocks}
\section{Introduction}
A natural way to improve the upper bound to sort permutations with prefix transpositions further is to increase the number of skipped or unvisited symbols of the permutation and use the sequence length algorithm to attain a double in the smallest number of regular or alternate greedy moves. In this chapter we shall improve the upper bound from $n-\log_{3}n$ to $n-\log_{2}n$ \cite{nair2021}. This is the best upper bound to date to sort permutations with prefix transpositions. In the previous chapter, we assumed that the $i^\text {th}$ interval has at most three skipped symbols. On increasing the number of skipped symbols in the $i^\text {th}$ interval, the number of alternate moves that we need to find would be very high; thus, the proof would be very lengthy and complicated. In this chapter, we shall follow a different approach to improve the upper bound using the concept of a \emph{block} along with the sequence length algorithm. 

\section{Block}
\begin{definition}
A \textbf{block} of a permutation $\pi$ with $n$ symbols is a sublist of $\pi$ with at least two elements and the additional property that when this sublist is sorted, then it becomes a substring of the sorted permutation $I_n$.
\end{definition}
For example, consider the permutation $\pi=(6,5,7,0,4,2,1,3,8)$ with nine symbols. Here $(6,5)$ is a block of $\pi$ with 2 symbols, $(6,5,7)$ and $(2,1,3)$ are blocks with three symbols, $(0,4,2,1,3)$ is a block containing five symbols and $(6,5,7,0,4,2,1,3)$ is a block with eight symbols. Note that a permutation $\pi$ with $n$ symbols is always a block with $n$ symbols.

\begin{lemma}
\label{lemmanew}
Let $C$ be a block of a permutation $\pi$ and $C \ne \pi$. If $C$ has some visited symbols then the last visited symbol in $C$ will be the smallest element in $C$, say $C_{min}$. $C$ has all the elements in the closed interval $[ C_{min}, C_{max}]$ where $C_{max}$ is the greatest element in $C$.
\end{lemma}
\begin{proof}
Let $\pi=(t_1,t_2,\dots,t_{j-1},\dots,s_j,t_j,\dots)$ and let $t_{j-1}$ be the last visited symbol in $C$, say $C=(\dots,t_{j-1},\dots)$. By the construction of sequence length algorithm, $t_{j-1}< c$ for every skipped symbol $c$ in the $j^\text {th}$ interval, and all symbols that lie to the left of $t_{j-1}$ in $\pi$ are greater than $t_{j-1}$. Since $t_{j-1}$ is the last visited symbol in $C$, all the symbols in $C$ that lie to the right of $t_{j-1}$ are skipped symbols in the $j^\text {th}$ interval. Hence $t_{j-1}$ is the minimum element in $C$. The second part of the lemma follows from the definition of a block.
\end{proof}

\section{Algorithm}
The basic principle used in the sequence length algorithm is to obtain a double preceded by singles in at most $\frac{7n}{8}$ moves and thus sort the permutation faster. If a sequence of greedy moves of the sequence length algorithm moves $n$ symbols, of which $\alpha n$ are visited and $(1-\alpha) n$ are skipped until a double is encountered ($0<\alpha \le 1$), then Chitturi \cite{Chitturi2015} has shown that the upper bound for sorting such a permutation is $n-\log_{(\frac{1}{1-\alpha })}n$. Further, by the recursive formula in Section \ref{Recursive formula}, the base of the logarithm is the ratio of the number of symbols moved by the number of symbols skipped. In our algorithm, which is a modified version of the sequence length algorithm, we shall use the concept of \emph{block} defined in the previous section along with some alternate and regular greedy moves to get a double faster. The $i^\text {th}$ interval is assumed to be the first interval in the permutation $\pi$ to have a skipped symbol. Then by Chitturi \cite{Chitturi2008}, $i\le 8$. 

\begin{observation}
\label{ob1}
Let $\pi=(t_1,t_2,\dots,s_j,t_j,\dots,s_k,t_k,\dots,x,y)$ be a permutation with at least $j-1$ skipped symbols in the $j^\text {th}$ interval. Then an upper bound to sort the permutation by prefix transpositions is $n-\log_{2}n$. 
\end{observation}
\begin{proof}
Suppose $\pi$ has at least $j-1$ skipped symbols in the $j^\text {th}$ interval. Then in $j-1$ regular greedy moves of the sequence length algorithm, we move at least $2(j-1)$ symbols among which $j-1$ are skipped. Hence by Lemma \ref{rf}, the upper bound is given by $n-\log_{(\frac{2(j-1)}{j-1})}n=n-\log_{2}n$. 
\end{proof}

\begin{lemma}
\label{lemma1}
Let $\pi=(t_1,\dots,s_2,t_2,\dots,x,y)$ be a permutation that contains a skipped symbol in the second interval. Then an upper bound for sorting $\pi$ by prefix transposition is $n-\log_{2}n$.
\end{lemma}
\begin{proof}
Consider the move first regular move
\begin{equation*}
\begin{gathered}
([t_1,\dots,s_2],t_2,\dots,t_1-1*,\dots)
\rightarrow (t_2,\dots,t_1-1,t_1,\dots,s_2,\dots)
\end{gathered}
\end{equation*}
which moves at least two symbols in a single move. Hence by Lemma \ref{rf}, we obtain an upper bound of $n-\log_{2}n$.
\end{proof}

To establish an upper bound of $n-\log_{2}n$ to sort permutations by prefix transpositions, due to Observation \ref{ob1}, Lemma \ref{lemma1} and the algorithm to sort $R_n$ \cite{Dias2002}, it is enough to prove that this bound is achieved when the number of skipped symbols in the $i^\text {th}$ interval is at most $i-2$, where $3\le i\le 8$. Hence for the rest of the discussions in this section, we shall assume that the $i^\text {th}$ interval contains at most $i-2$ skipped symbols and $3\le i\le 8$.

\begin{observation}
\label{ob2}
Let $\pi=(t_1,t_2,\dots,s_j,t_j,\dots,s_k,t_k,\dots,x,y)$ be a permutation with exactly $j-2$ skipped symbols in the $j^\text {th}$ interval. If the $(j+1)^\text {th}$ interval has more than one skipped symbol, then an upper bound to sort the permutation by prefix transpositions is $n-\log_{2}n$. 
\end{observation}
\begin{proof}
Suppose $\pi$ has $j-2$ skipped symbols in the $j^\text {th}$ interval and at least two skipped symbols in the $(j+1)^\text {th}$ interval. Then in $j$ regular greedy moves of the sequence length algorithm, we move at least $j+(j-2)+2=2j$ symbols among which $j$ are skipped. Hence by Lemma \ref{rf}, the upper bound is given by $n-\log_{(\frac{2j}{j})}n=n-\log_{2}n$. 
\end{proof}

\begin{lemma}
\label{lemma2}
Let $s$ be any skipped symbol in the $m^\text {th}$ interval of permutation $\pi$. If $s+1$ lies to the right of $t_m$, then $n-\log_{2}n$ is an upper bound for sorting $\pi$.
\end{lemma}
\begin{proof}
Consider a permutation of the form $\pi=(t_1,t_2,\dots,t_{m-1},\dots,s,\dots,t_m,\dots,s+1,\dots)$. The alternate move

\begin{equation*}
\begin{gathered}
([t_1,t_2,\dots,t_{m-1},\dots,s],\dots,t_m,\dots,*s+1,\dots)\\
\rightarrow (\dots,t_m,\dots,t_1,\dots,t_{m-1},\dots,s,s+1,\dots)\\[3ex]
\end{gathered}
\end{equation*}
moves at least $m$ symbols of which $m-1$ are skipped. Let $m^\text {th}$ interval have $j$ skipped symbols, note that $j \le m-2$ (Otherwise Observation \ref{ob1} gives a bound of $n -\log_2 n$) and each of these skipped symbols is greater than $t_{m-1}$. Let $s_k$ be the first symbol in the $m^\text {th}$ interval after the first move. We repeat applicable alternate moves of the forms shown below depending on whether $s_k-1$ is to the left or right of $t_m$. Thus, we move all the symbols that are to the left of $t_m$ to its right in at most $j-1$ moves.\\
1: $([s_k,\dots],(s_k-1),*\dots,s_m,t_m,\dots)
\rightarrow ((s_k-1),s_k,\dots,s_m,t_m,\dots)$\\
2: $([s_k,\dots,s_m],t_m,\dots,(s_k-1),* \ldots)$
$\rightarrow (t_m,\dots,(s_k-1),s_k,\dots,s_m,\dots)$\\
Thus, in $j$ moves, $m+(j-1)$ symbols are moved of which $m-1$ are skipped. This produces an upper bound of $n-\log_{(\frac{m+j-1}{m-1})}n$ which maximises to $n-\log_{(\frac{2m-3}{m-1})}n =n-\log_{(2-(\frac{1}{m-1}))}n$ when $j = m-2$. Hence the upper bound in this case is less than $n-\log_{2}n$.
\end{proof}

\begin{lemma}
\label{lemma3}
Let $\pi$ be a permutation with $n$ symbols, $C$ be a block in $\pi$ with $k$ symbols, where $k<n$. If $t_1$ belongs to $C$, a skipped symbol $c_0$ succeeds $C$ in $\pi$ and there is at least one interval in $\pi$ after $C$ that contains a skipped symbol, then an upper bound for sorting the permutation with prefix transpositions is $n-\log_{2}n$.
\end{lemma}
\begin{proof}
Suppose that $C=(t_1,t_2,\dots,t_j,\dots)$ is a proper sublist of $\pi$ that forms a block with $t_j$ being the last visited symbol in $C$, note that by Lemma \ref{lemmanew}, $t_j$ is the smallest number in $C$. Then $\pi=((t_1,t_2,\dots,t_j,\dots),c_0,c_1,\dots,c_l,t_{j+1},\dots,x,y)$, where $c_0,c_1,\dots,c_l$ are skipped symbols in the $(j+1)^\text {th}$ interval. Let $c_{m^\star}=\text{max}\{c_0,c_1,\dots,c_l\}$. If each $c_m+1$ lies to the left of $t_{j+1}$ in $\pi$ for $m=0,1,\dots,l$, then $c_{m^\star}+1$ lies in the block. This is possible only if $c_{m^\star}+1$ is the smallest number $t_j$ in the block by Lemma \ref{lemmanew}. Then by the sequence length algorithm, $c_{m^\star}$ is a visited symbol in $\pi$, a contradiction. So, there is at least one $c_m, 0\le m \le l$, such that $c_m+1$ lies to the right of $t_{j+1}$. Hence, by Lemma \ref{lemma2}, $n-\log_{2}n$ is an upper bound.
\end {proof}

\begin{lemma}
\label{lemma4}
Let $\pi$ be a permutation with $n$ symbols and $C$ be a block in $\pi$ with $k_1$ symbols, where $k_1<n$. If $t_1$ belongs to the block $C$, a visited symbol succeeds $C$ in $\pi$ and there is at least one interval in $\pi$ after $C$ that contains a skipped symbol, then an upper bound for sorting the permutation with prefix transpositions is $n-\log_{2}n$.
\end{lemma}
\begin{proof}
Let $C=(t_1,t_2,\dots,t_j,\dots)$ be a block with $t_j$ being the last visited symbol in $C$. Then $\pi=((t_1,t_2,\dots,t_j,\dots),t_{j+1},\dots,x,y)$. By definition, a block has at least 2 elements. Thus, if $j=1$ then $\pi$ has a skipped symbol in the second interval. So, by Lemma \ref{lemma1}, $n-\log_{2}n$ is an upper bound. We assume that $j\ge2$. 

If $t_{j+1}= t_j-1$, then we consider the new block including $t_{j+1}$ and restart the proof again if the element after the new block is visited. Note that this can happen only up to eight times. Otherwise, ${R_8}$ is present. If the element after the new block is a skipped symbol, then by Lemma \ref{lemma3}, $n-\log_{2}n$ is an upper bound.

Suppose $t_{j+1}\ne t_j-1$, then $ t_{j+1}= t_j-l $, for some $l>1$. Further, the symbols $(t_j-l+1)$ and $t_j-1$ are not in $C$ and hence they lie to the right of $ t_{j+1}$. Let $\beta$ be the symbol next to $t_{j+1}$ in $\pi$. Clearly $\beta \ne (t_j-l+1)$ because this would form an adjacency in $\pi$, which is not possible as the permutation we consider is reduced. Hence $\pi=((t_1,t_2,\dots,t_j,\dots),t_j-l,\beta,\dots,(t_j-l+1),\dots)$.\\

\emph{Case \ref{lemma4}.1}: If $\beta$ is a visited symbol ($\beta=t_{j+2}$), then the move
\begin{equation*}
\begin{gathered}
([(t_1,t_2,\dots,t_j,\dots),t_j-l],t_{j+2},\dots,*(t_j-l+1),\dots)\\
\rightarrow(t_{j+2},\dots,(t_1,t_2,\dots,t_j,\dots),t_j-l,(t_j-l+1),\dots)\\[3ex]
\end{gathered}
\end{equation*}
moves at least $j+1$ symbols of which $j$ are skipped. This produces an upper bound of $n-\log_{(\frac{j+1}{j})}n$ which maximises the base of the logarithm to $\frac{3}{2}$ when $j=2$. Thus $n-\log_{(\frac{3}{2})}n$ is an upper bound.\\

\emph{Case \ref{lemma4}.2}: If $\beta$ is a skipped symbol (say $\beta=c_1$), then the permutation is given by $\pi=((t_1,t_2,\dots,t_j,\dots),t_j-l,c_1,\dots,c_k,t_{j+2},\dots)$ where $c_1,\dots,c_k$ are skipped symbols in the $(j+2)^\text {th}$ interval. If $c_m+1$ lies to the right of $t_{j+2}$ for at least one $m =1,2,\dots,k$, then by Lemma \ref{lemma2}, $n-\log_{2}n$ is an upper bound.

Suppose that $c_m+1$ lies to the left of $t_{j+2}$ for all $m =1,2,\dots,k$. By the sequence length algorithm, every skipped symbol in the $(j+2)^\text {th}$ interval is greater than $t_j-l$. Let $c_M=max\{c_1,\dots,c_k\}$. Then $c_M+1$ is an element in $C$ and by Lemma \ref{lemmanew}, $c_M+1=t_j \Rightarrow c_M=t_j-1$. Further, for all the other skipped symbols $c_m$, except for $m=M$, $c_m+1$ lies in the $(j+2)^\text {th}$ interval by Lemma \ref{lemmanew}. Hence all the skipped symbols are necessarily consecutive symbols from $t_j-k$ to $t_j-1$.\\ 

\noindent\textbf{Statement 1:} $\pi=((t_1,t_2,\dots,t_j,\dots),t_j-l,c_1,\dots,c_k,t_{j+2},\dots)$, where $c_1,\dots,c_k$ are consecutive symbols from $t_j-k$ to $t_j-1$.
Note that $t_j-l=t_{j+1}$ is a visited symbol and $c_i>t_j-l$ for $1\le i \le k$ as they are skipped elements. So $t_j-k > t_j-l$ which gives us $k<l$. Here we shall consider two cases: (1) $k<l-1$ and (2) $k=l-1$\\

\emph{Case \ref{lemma4}.2.1}: Suppose that $k<l-1$, then $(t_j-l+1)$ and $(t_j-k-1)$ lies to the right of $t_{j+2}$. Note that $(t_j-l+1)=(t_j-k-1)$ if $k=l-2$. Here $\pi=((t_1,t_2,\dots,t_j,\dots),t_j-l,c_1,\dots,c_k,t_{j+2},\dots,(t_j-l+1),\dots)$, where the symbols $c_1, c_2,\dots,c_k$ are a rearrangement of $(t_j-k,t_j-k+1, \dots, t_j-1)$. Further by Observation \ref{ob1}, we need to only consider the case when the number of skipped symbols in the $(j+2)^\text {th}$ interval is at most $j$, hence $k \le j$.

(i) Let $k < j$. Consider the following sequence of alternate moves. The first move moves at least $j+1$ symbols of which $j$ are skipped.

\begin{equation*}
\begin{gathered}
([(t_1,t_2,\dots,t_j,\dots),t_j-l],c_1,\dots,c_k,t_{j+2},\dots,*(t_j-l+1),\dots)\\
\rightarrow(c_1,\dots,c_k,t_{j+2},\dots,(t_1,t_2,\dots,t_j,\dots),t_j-l,(t_j-l+1),\dots)\\[3ex]
\end{gathered}
\end{equation*}
Next, we make the following move repeatedly until $t_j-k$ becomes the first symbol of $\pi$. Note that this is possible by Statement 1 in $k-1$ moves.
\begin{equation*}
\begin{gathered}
([c_m,\dots],(c_m-1),*\dots,t_{j+2},\dots) \rightarrow((c_m-1),c_m,\dots,t_{j+2},\dots)
\end{gathered}
\end{equation*}
The last move in the sequence is

\begin{equation*}
\begin{gathered}
([t_j-k,\dots],t_{j+2},\dots,(t_j-k-1),*\dots)\\
\rightarrow(t_{j+2},\dots,(t_j-k-1),t_j-k,\dots)\\[3ex]
\end{gathered}
\end{equation*}
Here at least $(j+1)+k$ symbols are moved in at most $k+1$ moves of which $j$ are skipped. This produces an upper bound of $n-\log_{(\frac{j+k+1}{j})}n \le n-\log_{2}n$ by Lemma \ref{rf}.

(ii) Let $k = j$. By Observation \ref{ob2}, the $(j+3)^\text {th}$ interval has at most one skipped symbol (say $c$). Then $\pi=((t_1,t_2,\dots,t_j,\dots),t_j-l,c_1,\dots,c_k,t_{j+2},c,t_{j+3}\dots)$.
If $(t_{j+2}+1)$ lies to the right of $t_{j+3}$, then in the following two moves we move $(j+1)+j+2=2j+3$ symbols of which $2j+1$ are skipped. So, by Lemma \ref{rf}, an upper bound is given by $n-\log_{(\frac{2j+3}{2j+1})}n=n-\log_{(1+\frac{2}{2j+1})}n$ which is less than $n-\log_{2}n$ since $j\ge2$.
\begin{equation*}
\begin{gathered}
([(t_1,t_2,\dots,t_j,\dots),t_j-l,c_1,\dots,c_k,t_{j+2}],c,t_{j+3},\dots,*(t_{j+2}+1),\dots)\\
\rightarrow(c,t_{j+3},\dots,(t_1,t_2,\dots,t_j,\dots),t_j-l,c_1,\dots,c_k,t_{j+2},(t_{j+2}+1),\dots)\\
([c],t_{j+3},\dots,*(c+1),\dots)
\rightarrow (t_{j+3},\dots,c,(c+1),\dots)\\[3ex]
\end{gathered}
\end{equation*}

If $(t_{j+2}+1)$ lies to the left of $t_{j+3}$, then $(t_{j+2}+1)=t_j-l$. Here the permutation becomes $\pi=((t_1,t_2,\dots,t_j,\dots),t_j-l,c_1,\dots,c_k,(t_j-l-1),c,t_{j+3},\dots)$. Consider the position of $(t_j-l+1)$ in $\pi$. In $\pi$, note that as shown in the beginning of this case $(t_j-l+1)$ should lie to the right of $t_j+2=(t_j-l-1)$.

If $c \ne (t_j-l+1)$, $\pi=((t_1,t_2,\dots,t_j,\dots),t_j-l,c_1,\dots,c_k,(t_j-l-1),c,t_{j+3},\dots,(t_j-l+1),\dots)$. In the following three moves we move $2j+3$ symbols of which $2j$ are skipped. So, by Lemma \ref{rf}, an upper bound is given by $n-\log_{(\frac{2j+3}{2j})}n=n-\log_{(1+\frac{3}{2j})}n$ which is less than $n-\log_{2}n$ since $j\ge2$. Note that if the skipped symbol $c$ does not exist then we will execute the first two alternate moves mentioned above and we would have moved $2j+2$ elements in two moves, giving us a bound of $n-\log_{(\frac{2j+2}{2j})}n=n-\log_{(1+\frac{1}{j})}n$ which is less than $n-\log_{2}n$ since $j\ge2$.
\begin{equation*}
\begin{gathered}
([(t_1,\dots,t_j,\dots),t_j-l],c_1,\dots,c_k,(t_j-l-1),c,t_{j+3},\dots,*(t_j-l+1),\dots)\\
\rightarrow (c_1,\dots,c_k,(t_j-l-1),c,t_{j+3},\dots,(t_1,\dots,t_j,\dots),t_j-l,(t_j-l+1),\dots)\\
([c_1,\dots,c_k,(t_j-l-1)],c,t_{j+3},\dots,(t_1,\dots,t_j,\dots),*t_j-l,(t_j-l+1),\dots)\\
\rightarrow (c,t_{j+3},\dots,(t_1,\dots,t_j,\dots),c_1,\dots,c_k,(t_j-l-1),t_j-l,(t_j-l+1),\dots)\\
([c],t_{j+3},\dots,*(c+1),\dots)
\rightarrow (t_{j+3},\dots,c,(c+1),\dots)\\[3ex]
\end{gathered}
\end{equation*}

If $c = (t_j-l+1)$ and $k<l-2$, then $(t_j-k-1)$ lies to the right of $t_{j+3}$ and $\pi=((t_1,t_2,\dots,t_j,\dots),t_j-l,c_1,\dots,c_k,(t_j-l-1),(t_j-l+1),t_{j+3},\dots,(t_j-k-1),\dots)$. we shall consider the following sequence of alternate moves. The first move 
\begin{equation*}
\begin{gathered}
([(t_1,\dots,t_j,\dots),t_j-l],c_1,\dots,(t_j-l-1),*(t_j-l+1),t_{j+3},\dots,(t_j-k-1),\dots)\\
\rightarrow (c_1,\dots,(t_j-l-1),(t_1,\dots,t_j,\dots),t_j-l,(t_j-l+1),t_{j+3},\dots,(t_j-k-1), \dots)\\
\end{gathered}
\end{equation*}
moves at least $j+1$ symbols of which $j$ are skipped. Next, we make the following move
\begin{equation*}
\begin{gathered}
([c_m,\dots],(c_m-1),*\dots,(t_j-l-1),\dots)\\
\rightarrow((c_m-1),c_m,\dots,(t_j-l-1),\dots)\\[3ex]
\end{gathered}
\end{equation*}
repeatedly until $t_j-k$ becomes the first symbol in $\pi$, noting that this is possible due to Statement 1. This can be attained in $k-1$ moves. The last move is 
\begin{equation*}
\begin{gathered}
([t_j-k,\dots,(t_j-l-1),\dots,(t_j-l+1)],t_{j+3},\dots,(t_j-k-1),*\dots)\\
\rightarrow(t_{j+3},\dots,(t_j-k-1),t_j-k,\dots)\\[3ex]
\end{gathered}
\end{equation*}
Here at least $(j+1)+k+2$ symbols are moved in at most $k+1$ moves of which $j+2$ are skipped. This produces an upper bound of $n-\log_{(\frac{j+k+3}{j+2})}n= n-\log_{(\frac{2j+3}{j+2})}n \le n-\log_{2}n$ by Lemma \ref{rf}. 

If $c = (t_j-l+1)$ and $k=l-2$, then $(t_j-k-1)=(t_j-l+1)$ and $\pi=((t_1,t_2,\dots,t_j,\dots),t_j-l,c_1,\dots,c_k,(t_j-l-1),(t_j-l+1),t_{j+3},\dots)$. Further all the symbols before $t_{j+3}$ in $\pi$ form a block with least element $(t_j-l-1)$, here we use Statement $1$. We will call this block $C_1$. Here we shall consider two cases:

\emph{Case} (a): If the initial block $C=(t_1,\dots,t_j,\dots)$ contains skipped symbols, then $C$ has at least $j+1$ symbols. Here we shall execute the sequence of moves as in the previous case where $c = (t_j-l+1)$ and $k<l-2$, until $t_j-k$ becomes the first symbol in $\pi$. The next couple of moves is 
\begin{equation*}
\begin{gathered}
([t_j-k,\dots],(t_j-l-1),\dots,(t_j-k-1)*,t_{j+3},\dots)\\
\rightarrow((t_j-l-1),\dots,(t_j-k-1),t_j-k,\dots,t_{j+3},\dots)\\
\end{gathered}
\end{equation*}
\begin{equation*}
\begin{gathered}
([(t_j-l-1),\dots,(t_j-k-1),t_j-k,\dots],t_{j+3},\dots,(t_j-l-2)*\dots)\\
\rightarrow (t_{j+3},\dots,(t_j-l-2),(t_j-l-1),\dots,(t_j-k-1),t_j-k,\dots)\\[3ex]
\end{gathered}
\end{equation*}
Here at least $(j+2)+k+2$ symbols are moved in at most $k+2$ moves of which $j+2$ are skipped. This produces an upper bound of $n-\log_{(\frac{j+k+4}{j+2})}n = n-\log_{2}n$ by Lemma \ref{rf}. 

\emph{Case} (b): Suppose that there are no skipped symbols in block $C$. Then we shall consider the block $C_1$, instead of block $C$, mentioned above, noting that the block $C_1$ has skipped symbols $c_1,\dots,c_k$. Further the visited symbol $t_{j+3}$ follows $C_1$ and hence either of the cases \ref{lemma4}.1 or \ref{lemma4}.2.1 - (except case (b)) or \ref{lemma4}.2.2- (except case(b)) applies giving us the required upper bound. We will prove Case \ref{lemma4}.2.2 below.

\emph{Case \ref{lemma4}.2.2}: Suppose $k=l-1$. Then $\pi=((t_1,t_2,\dots,t_j,\dots,t_j-l,c_1,\dots,(t_j-l+1),\dots),t_{j+2},\dots)$, where $(t_j-l+1)$ is a skipped symbol in the $(j+2)^\text {th}$ interval due to Statement 1. Also, by Lemma \ref{lemmanew} and Statement 1, all the symbols to the left of $t_{j+2}$ form a block with least symbol $t_j-l$. 

If $k < j$, consider the following sequence of alternate moves. The first move 

\begin{equation*}
\begin{gathered}
([(t_1,t_2,\dots,t_j,\dots,t_j-l],c_1,\dots,*(t_j-l+1),\dots),t_{j+2},\dots)\\
\rightarrow(c_1,\dots,(t_1,t_2,\dots,t_j,\dots,t_j-l,(t_j-l+1),\dots),t_{j+2},\dots)
\end{gathered}
\end{equation*}
\\moves at least $j+1$ symbols of which $j$ are skipped. Next, we make the move
\begin{equation*}
\begin{gathered}
([c_m,\dots],(c_m-1),*\dots,t_{j+2},\dots)
\rightarrow((c_m-1),c_m,\dots,t_{j+2},\dots)
\end{gathered}
\end{equation*}
repeatedly until $(t_j-l+2)$ becomes the first symbol in $\pi$, noting that these moves are possible by Statement 1. This can be attained in $k-2$ moves as $(t_j-l+1)$ would not be the first symbol in any of these moves. The next two moves in the sequence are 
\begin{equation*}
\begin{gathered}
([(t_j-l+2),\dots],t_j-l,(t_j-l+1),*\dots,t_{j+2},\dots)\\
\rightarrow(t_j-l,(t_j-l+1),(t_j-l+2),\dots,t_{j+2},\dots)\\
([t_j-l,(t_j-l+1),(t_j-l+2),\dots],t_{j+2},\dots,(t_j-l-1),*\dots)\\
\rightarrow t_{j+2},\dots,(t_j-l-1),t_j-l,(t_j-l+1),(t_j-l+2),\dots)\\[3ex]
\end{gathered}
\end{equation*}
Here at least $(j+1)+k$ symbols are moved in at most $k+1$ moves of which $j$ are skipped. This produces an upper bound of $n-\log_{(\frac{j+k+1}{j})}n \le n-\log_{2}n$ by Lemma \ref{rf}. Note that the same moves work even if $ t_j+2=t_j-l-1$. 

If $k = j$, then there are $j$ skipped symbols in the $(j+2)^\text {th}$ interval. Hence by Observation \ref{ob2}, the $(j+3)^\text {th}$ interval has at most one skipped symbol (say $c$). Then $\pi=((t_1,t_2,\dots,t_j,\dots),t_j-l,c_1,\dots,(t_j-l+1),\dots,t_{j+2},c,t_{j+3}\dots)$.

If $(t_{j+2}+1)$ lies to the right of $t_{j+3}$, then in the following two moves 
\begin{equation*}
\begin{gathered}
([(t_1,\dots,t_j,\dots),t_j-l,c_1,\dots,(t_j-l+1),\dots,t_{j+2}],c,t_{j+3},\dots,*(t_{j+2}+1),\dots)\\
\rightarrow(c,t_{j+3},\dots,(t_1,\dots,t_j,\dots),t_j-l,c_1,\dots,(t_j-l+1),\dots,t_{j+2},(t_{j+2}+1),\dots)\\
([c],t_{j+3},\dots,*c+1,\dots)
\rightarrow (t_{j+3},\dots,c,c+1,\dots)\\[3ex]
\end{gathered}
\end{equation*}
we move $(j+1)+j+2=2j+3$ symbols of which $2j+1$ are skipped. So, by Lemma \ref{rf}, an upper bound is given by $n-\log_{(\frac{2j+3}{2j+1})}n=n-\log_{(1+\frac{2}{2j+1})}n$ which is less than $n-\log_{2}n$ since $j\ge2$. Note that if $c$ does not exist, then the first alternate move mentioned above should suffice giving us a bound of $n-\log_{(\frac{2j+2}{2j+1})}n=n-\log_{(1+\frac{1}{2j+1})}n$ which is less than $n-\log_{2}n$ since $j\ge2$.

If $(t_{j+2}+1)$ lies to the left of $t_{j+2}$, then $(t_{j+2}+1)=t_j-l$ and hence $((t_1,t_2,\dots,t_j,\dots),t_j-l,c_1,\dots,(t_j-l+1),\dots,t_{j+2})$ is a block, by Statement 1, using the fact that $k=j=l-1$. We will call this block $C_1$. Here we shall consider two cases:

\emph{Case} (a): If the initial block $C=(t_1,\dots,t_j,\dots)$ contains skipped symbols, then $C$ has at least $j+1$ symbols. Here, we consider the following sequence of alternate moves where the first move
\begin{equation*}
\begin{gathered}
([(t_1,t_2,\dots,t_j,\dots,t_j-l],c_1,\dots,*(t_j-l+1),\dots),t_{j+2},\dots)\\
\rightarrow(c_1,\dots,(t_1,t_2,\dots,t_j,\dots,t_j-l,(t_j-l+1),\dots),t_{j+2},\dots)\\[3ex]
\end{gathered}
\end{equation*}
moves at least $j+2$ symbols of which $j+1$ are skipped. Next, we make the following move
\begin{equation*}
\begin{gathered}
([c_m,\dots],(c_m-1),*\dots,t_{j+2},\dots)
\rightarrow((c_m-1),c_m,\dots,t_{j+2},\dots)
\end{gathered}
\end{equation*}
repeatedly until $(t_j-l+2)$ becomes the first symbol in $\pi$, noting that these moves are possible by Statement 1. This can be attained in at most $j-2$ moves as $(t_j-l+1)$ would not be the first symbol in any of these moves. The next two moves in the sequence are given by
\begin{equation*}
\begin{gathered}
([(t_j-l+2),\dots],t_j-l,(t_j-l+1),*\dots,t_{j+2},\dots)\\
\rightarrow(t_j-l,(t_j-l+1),(t_j-l+2),\dots,t_{j+2},\dots)\\
([t_j-l,(t_j-l+1),(t_j-l+2),\dots],t_{j+2},\dots,(t_j-l-1),*\dots)\\
\rightarrow t_{j+2},\dots,(t_j-l-1),t_j-l,(t_j-l+1),(t_j-l+2),\dots)\\[3ex]
\end{gathered}
\end{equation*}
Here at least $(j+2)+j=2j+2$ symbols are moved in at most $j+1$ moves of which $j+1$ are skipped. This produces an upper bound of $n-\log_{(\frac{2j+2}{j+1})}n=n-\log_{2}n$ by Lemma \ref{rf}.

\emph{Case} (b): Suppose that there are no skipped symbols in block $C$. Then we shall consider the block $C_1$, instead of block $C$, mentioned above, note that the block $C_1$ has skipped symbols $c_1,\dots,c_k$. If a skipped element follows $C_1$ in $\pi$ then we get $n-\log_{2}n$ bound by Lemma \ref{lemma3}. If a visited element follows $C_1$ then either of the cases \ref{lemma4}.1 or \ref{lemma4}.2.1 (except case (b)) or \ref{lemma4}.2.2 (except case (b)) applies giving us the required upper bound.
\end{proof}

\begin{lemma}
\label{lemma5}
Let $\pi$ be a permutation with $n$ symbols, in which $(t_1,t_2)$ is not a block and the only block containing the first symbol is whole of $\pi$ .Then the upper bound for sorting $\pi$ using prefix transpositions is $n-\log_{2}n$.
\end{lemma}
\begin{proof}
Suppose that $\pi$ is a permutation in which $t_1$ is not contained in a block with less than $n$ symbols. Then $t_2 \ne t_1-1$. If the second symbol in $\pi$ is a skipped symbol, then by Lemma \ref{lemma1}, $n-\log_{2}n$ is an upper bound. Hence, we shall assume $\pi=(t_1,t_1-k,\dots,x,y)$ where $k>1$. Here we shall consider the third symbol $\alpha$ in $\pi$.

If $\alpha$ is a visited symbol ($\alpha=t_3$) then $(t_1-k+1)$ being a skipped symbol lies to the right of $t_3$. Here $\pi=(t_1,t_1-k,t_3,\dots,(t_1-k+1),\dots,x,y)$ and in one move

\begin{equation*}
\begin{gathered}
([t_1,t_1-k],t_3,\dots,*(t_1-k+1),\dots,x,y)\\
\rightarrow(t_3,\dots,t_1,t_1-k,(t_1-k+1),\dots,x,y)\\[3ex]
\end{gathered}
\end{equation*}
two symbols are moved and one skipped. Hence the upper bound for sorting $\pi$ is $n-\log_{2}n$.

If $\alpha$ is a skipped symbol ($\alpha=s_3$), then by Observation \ref{ob1}, it is the only skipped symbol in the third interval. By the sequence length algorithm, $t_1-k=s_3-l$ for some $l>1$. If $s_3+1$ lies to the right of $t_3$, then by Lemma \ref{lemma2}, the upper bound is less than $n-\log_{2}n$. Hence, we shall only consider the case when $\pi=(s_3+1,s_3-l,s_3,t_3,\dots,x,y)$, where $l>1$. Let $\beta$ be the symbol that follows $t_3$ in $\pi$.

\emph{Case \ref{lemma5}.1}: When $\beta$ is a visited symbol ($\beta=t_4$, then $s_3-1$ lies to the right of $t_4$. In the following two alternate moves
\begin{equation*}
\begin{gathered}
([s_3+1,s_3-l],s_3,*t_3,t_4,\dots,s_3-1,\dots)\\
\rightarrow (s_3,s_3+1,s_3-l,t_3,t_4,\dots,s_3-1,\dots)\\
([s_3,s_3+1,s_3-l,t_3],t_4\dots,s_3-1,*\dots)\\
\rightarrow(t_4,\dots,s_3-1,s_3,s_3+1,s_3-l,t_3,\dots)\\[3ex]
\end{gathered}
\end{equation*}
four symbols are moved of which two are skipped. Hence the upper bound in this case is $n-\log_{2}n$.

\emph{Case \ref{lemma5}.2}: When $\beta$ is a skipped symbol then, by Observation \ref{ob2}, $\beta=s_4$ is the only skipped symbol in the fourth interval. Then permutation $\pi$ equals $(s_3+1,s_3-l,s_3,t_3,s_4,t_4,\dots,x,y)$. If $s_4+1$ lies to the right of $t_4$, by Lemma \ref{lemma2}, the upper bound is less than $n-\log_{2}n$. 

If $t_3+1$ lies to the right of $t_4$, then the following two moves 
\begin{equation*}
\begin{gathered}
([s_3+1,s_3-l,s_3,t_3],s_4,t_4,\dots,*t_3+1,\dots)\\
\rightarrow(s_4,t_4,\dots,s_3+1,s_3-l,s_3,t_3,t_3+1,\dots)\\
([s_4],t_4,\dots,*s_4+1,\dots)
\rightarrow(t_4,\dots,s_4,s_4+1,\dots)\\[3ex]
\end{gathered}
\end{equation*}
will move five symbols of which three are skipped. Hence the upper bound in this case is $n-\log_{(\frac{5}{3})}n$.

Now we shall consider the case when $s_4+1$ and $t_3+1$ lies to the left of $t_3$. This is possible only when $s_4+1=s_3$ and $t_3+1=s_3-l$ in $\pi$. Then $\pi=(s_3+1,s_3-l,s_3,(s_3-l-1),s_3-1,t_4,\dots,x,y)$. Here we repeat cases \ref{lemma5}.1 and \ref{lemma5}.2 in the lemma by assuming $\beta$ to be the symbol that follows $t_4$ in $\pi$. Using similar alternate moves we shall obtain an upper bound of $n-\log_{2}n$ for all cases except when $\pi=(s_3+1,s_3-l,s_3,(s_3-l-1),s_3-1,(s_3-l-2),s_3-2,t_5,\dots,x,y)$. Repeating the same argument $l$ times for increasing values of $i$, the subscript of $t_i$, we get the following.\\

\noindent\textbf{Statement 2} If $\pi$ is a permutation in which $(t_1,t_2)$ is not a block, then we obtain an upper bound of $n-\log_{2}n$ unless the permutation in any one of the forms given below:
\begin{equation*}
\begin{gathered}
\pi=(s_3+1,s_3-l,s_3,\dots,(s_3-2l+1),(s_3-l+1),s_3-2l,\dots,x,y).\\
\pi=(s_3+1,s_3-l,s_3,\dots,(s_3-2l+1),(s_3-l+1),\dots,x,y).\\
\end{gathered}
\end{equation*}
Note that the elements from $s_3+1$ up to $s_3-2l$ form a block, say $C_1$, in the first case, with $s_3-2l$ the smallest element and $s_3+1$ the largest element. In the second case, elements from $s_3+1$ up to $s_3-l+1$ form a block, say $C_2$, with $s_3-2l+1$ the smallest element and $s_3+1$ the largest element. Since the smallest block containing $t_1=s_3+1$ is the entire permutation $\pi$, we get $\pi=C_1$ or $\pi=C_2$. We shall consider these two cases below:

\emph{Case} (i): If $\pi=(s_3+1,s_3-l,s_3,(s_3-l-1),s_3-1,(s_3-l-2),\dots,(s_3-l+1),(s_3-2l))$. Then $\pi$ contains $2l+2$ symbols. We relabel the symbols in odd positions as visited and the symbols in even position as skipped. Then we do the regular greedy moves of the sequence length algorithm until $(s_3-l+2)$ becomes the first symbol. This would be accomplished in $l-2$ moves since every alternate symbol starting from first symbol $s_3+1$ are in consecutive decreasing order and we skip exactly one symbol in each move. The first two moves are given below
\begin{equation*}
\begin{gathered}
([s_3+1,s_3-l],s_3,*(s_3-l-1),s_3-1,(s_3-l-2),\dots,(s_3-l+1),s_3-2l)\\
\rightarrow(s_3,s_3+1,s_3-l,(s_3-l-1),s_3-1,(s_3-l-2),\dots,(s_3-l+1),s_3-2l)\\
([s_3,s_3+1,s_3-l,(s_3-l-1)],s_3-1,*(s_3-l-2),\dots,(s_3-l+1),s_3-2l)\\
\rightarrow(s_3-1,s_3,s_3+1,s_3-l,(s_3-l-1),(s_3-l-2),\dots,(s_3-l+1),s_3-2l)\\[3ex]
\end{gathered}
\end{equation*}
When $(s_3-l+2)$ becomes the first symbol in $\pi$, we execute the move
\begin{equation*}
\begin{gathered}
([(s_3-l+2),\dots,(s_3-2l-1)],\dots,(s_3-l+1),*s_3-2l)\\
\rightarrow(\dots,(s_3-l+1),(s_3-l+2),\dots,(s_3-2l-1),s_3-2l)\\[3ex]
\end{gathered}
\end{equation*}
which is a double. In the following sequence of moves 
we create $l$ adjacencies in $l-1$ moves and skip $l$ symbols. Hence, an upper bound in this case is less than $n-\log_{2}n$.\\

\emph{Case} (ii): If $\pi=(s_3+1,s_3-l,s_3,(s_3-l-1),s_3-1,(s_3-l-2),\dots,(s_3-2l+1),(s_3-l+1))$, we proceed as in case (i) until $(s_3-l+2)$ becomes the first symbol in $\pi$, the next move is given by 

\begin{equation*}
\begin{gathered}
([(s_3-l+2),..s_3+1],\dots,(s_3-l+1),*)\\
\rightarrow(\dots,(s_3-l+1),(s_3-l+2),\dots,s_3+1)\\[3ex]
\end{gathered}
\end{equation*}
which is a double as it creates an adjacency and places $s_3+1$, the largest element in $\pi$ in the last position (creating another adjacency).
So, the sequence of moves shown above forms $l$ adjacencies in $l-1$ moves and skips $l$ symbols yielding an upper bound of less than $n-\log_{2}n$.
\end{proof}

\begin{theorem}
\label{theorem1}
An upper bound for sorting permutations with $n$ symbols using prefix transposition is $n-\log_{2}n$.
\end{theorem}
\begin{proof}
Let $\pi=(t_1,t_2,\dots,s_i,t_i,\dots,x,y)$ be a permutation with $n$ symbols. If $(t_1,t_2)$ forms a block in $\pi$, then by Lemma \ref{lemma3} and \ref{lemma4}, an upper bound of $n-\log_{2}n$ holds. Consider the case when $(t_1,t_2)$ is not a block, then by Statement 2, we get an upper bound of $n-\log_{2}n$ unless $\pi=(s_3+1,s_3-l,s_3,\dots,(s_3-2l+1),(s_3-l+1),s_3-2l,\dots,x,y)$ or $\pi=(s_3+1,s_3-l,s_3\dots,(s_3-2l+1),(s_3-l+1),\dots,x,y)$. We know that $C_1=(s_3+1,s_3-l,s_3,\dots,(s_3-2l+1),(s_3-l+1),s_3-2l)$ is a block in the first case and $C_2=(s_3+1,s_3-l,s_3,\dots,(s_3-2l+1),(s_3-l+1))$ is a block in the second case. Let $C$ be $C_1$ or $C_2$ depending on whether $\pi$ is of the first or the second case. If there is at least one interval in $\pi$ after $C$, that contains a skipped symbol, then by Lemmas \ref{lemma3} and \ref{lemma4}, $n-\log_{2}n$ is an upper bound to sort $\pi$ using prefix transpositions. Now we shall consider the remaining case where there is no interval after $C$ that contains a skipped symbol. This case can be partitioned into two sub-cases.

\emph{Sub-case} 1: If $\pi$ forms a single block, that is $\pi$ is either $C_1$ or $C_2$ then Lemma \ref{lemma5} establishes the upper bound.

\emph{Sub-case} 2: Either $\pi=(s_3+1,s_3-l,s_3,\dots,(s_3-2l+1),(s_3-l+1),s_3-2l,\dots,x,y)$ or $\pi=(s_3+1,s_3-l,s_3\dots,(s_3-2l+1),(s_3-l+1),\ldots x, y)$.
That is, $\pi= (C_1, \ldots x, y)$ or $\pi= (C_2, \ldots x, y)$. $C$ denotes either $C_1$ or $C_2$ depending on the context. Here all the symbols between $C$ and $y$ are visited symbols, with no skipped symbols in between them. Thus, the elements between $C$ and y will be decreasing elements of the form $(j, j-1, j-2, \dots)$. 
Let $C$ contain $n_1$ symbols and let $n_2=n-n_1$, then $\pi=(C, R_{n_2})$, where $R_{n_2}$ is the reverse permutation with $n_2$ elements.
$C$ is a permutation, a sub-permutation of $\pi$ with $n_1$ symbols which satisfies the conditions in Lemma \ref{lemma5} and hence can be sorted in at most $n_1-\log_{2}n_1$ moves.
After sorting $C$, $\pi$ reduces to $R_{n_2+1}$. 
This resultant permutation can be sorted with Dias and Meidanis sequence \cite{Dias2002} in at most $\frac{3(n_2+1)}{4} +O(1)$ moves. Thus, an upper bound to sort $\pi$ is $(n_1-\log_{2}n_1+\frac{3n_2+3}{4}) =n_1+n_2-(\log_{2}n_1+\frac{n_2}{4}-\frac{3}{4})<n_1+n_2-\log_{2}(n_1+n_2)=n-\log_{2}n$. Thus, the theorem follows.
\end{proof}

\chapter{Summary, Conclusions, and Scope for Further Research}

Prefix transpositions were introduced and studied by Dias, and Meidanis \cite{Dias2002} in 2002. In this introductory article, they provided an algorithm to sort the reverse permutation $R_n$ in $n-\lfloor{\frac{n}{4}\rfloor}$ prefix transpositions. As $R_n$ is considered to be the hardest permutation to sort, it is conjectured that $\frac{3n}{4}$ is an upper bound to sort a permutation with $n$ symbols. Chitturi et al. \cite{Chitturi2008} in 2008 defined the sequence length algorithm that improved the upper bound from $n-1$ to $n-\log_8n$. This thesis uses the sequence length algorithm to improve the upper bound to sort permutations with prefix transposition to $n-\log_3n$ and then to $n-\log_2n$. 

In chapters 3 and 4, we considered the first interval to contain unvisited symbols ($s_i$ is the last unvisited symbol in the interval) and restricted the number of unvisited symbols in it. Then we introduced some alternate moves in the sequence length algorithm depending on the position of the symbol $s_i+1$ to improve the upper bound. In Chapter 5, we defined a block in a permutation and defined alternate moves to get an upper bound of $n-\log_2n$ by considering the following cases:
\begin{itemize}
\item if a skipped symbol succeeds a block C containing the first symbol, and there are skipped symbols in $\pi$ after C.
\item if a visited symbol succeeds a block C containing the first symbol and there are skipped symbols in $\pi$ after C.
\item when $t_1$ and $t_2$ does not form a block.
\item if the only block in C with the first symbol $t_1$ is the whole of $\pi$.
\end{itemize}

Due to Lemma \ref{rf}, the recursive formula, it may seem that by adding additional alternate moves, one can achieve an upper bound of $n-\log_{(1+\epsilon)}n;(\epsilon >0)$ using the sequence length algorithm. But this is not the case. Consider a permutation $\pi$ such that every interval in the permutation has exactly one skipped symbol, and for each skipped symbol $s$, $s+1$ lies to the left of $s$ in $\pi$. Then the permutation can be written as either $\pi=(s_3+1,s_3-l,s_3,\dots,(s_3-2l+1),(s_3-l+1),s_3-2l)$ or $\pi=(s_3+1,s_3-l,s_3,\dots,(s_3-2l+1),(s_3-l+1))$. From cases (i) and (ii) in Lemma \ref{lemma5}, we proved that $n-\log_2n$ is an upper bound to sort $\pi$ by prefix transpositions using the sequence length algorithm and alternate moves. The only three prefix transpositions on $\pi$ that guarantee a single are given by moving $s_3+1$ after $s_3$ or moving $[s_3+1,s_3-l]$ after $s_3$ or moving $[s_3+1,s_3-l]$ in front of $(s_3-l+1)$. In all these cases, we are moving a maximum of two symbols in one move, giving us an upper bound of $n-\log_2n$. As there are no other moves possible, this is the best upper bound we can get using the sequence length algorithm and the alternate moves. Hence, we see that this algorithm cannot be used to improve the upper bound further. So even though there is a large gap between the present upper bound of $n-\log_2n$ and the conjectured upper bound of $\frac{3n}{4}$, we presume that a completely new technique similar to $R_n$ needs to be introduced for further improvement of this upper bound.

\newpage
\addcontentsline{toc}{chapter}{\hspace*{.53cm}References}
\bibliography{prefix}

\begin{thebibliography}{10}

\bibitem{Akers1989}
Sheldon~B. Akers and Balakrishnan Krishnamurthy.
\newblock A group-theoretic model for symmetric interconnection networks.
\newblock {\em IEEE transactions on Computers}, 38(4):555--566, 1989.

\bibitem{alexandrino2020complexity}
Alexsandro~Oliveira Alexandrino, Andre~Rodrigues Oliveira, Ulisses Dias, and
  Zanoni Dias.
\newblock On the complexity of some variations of sorting by transpositions.
\newblock {\em Journal of Universal Computer Science}, 26(9):1076--1094, 2020.

\bibitem{Bafna1998}
Vineet Bafna and Pavel~A Pevzner.
\newblock Sorting by transpositions.
\newblock {\em SIAM Journal on Discrete Mathematics}, 11(2):224--240, 1998.

\bibitem{bender2008improved}
Michael~A Bender, Dongdong Ge, Simai He, Haodong Hu, Ron~Y Pinter, Steven
  Skiena, and Firas Swidan.
\newblock Improved bounds on sorting by length-weighted reversals.
\newblock {\em Journal of Computer and System Sciences}, 74(5):744--774, 2008.

\bibitem{berman20021}
Piotr Berman, Sridhar Hannenhalli, and Marek Karpinski.
\newblock 1.375-approximation algorithm for sorting by reversals.
\newblock In {\em European Symposium on Algorithms}, pages 200--210. Springer,
  2002.

\bibitem{Bulteau2012}
Laurent Bulteau, Guillaume Fertin, and Irena Rusu.
\newblock Sorting by transpositions is difficult.
\newblock {\em SIAM Journal on Discrete Mathematics}, 26(3):1148--1180, 2012.

\bibitem{Caprara1997}
Alberto Caprara.
\newblock Sorting by reversals is difficult.
\newblock In {\em Proceedings of the first annual international conference on
  Computational molecular biology}, pages 75--83, 1997.

\bibitem{chen2010genomic}
Jian-Min Chen, David~N Cooper, Claude F{\'e}rec, Hildegard Kehrer-Sawatzki, and
  George~P Patrinos.
\newblock Genomic rearrangements in inherited disease and cancer.
\newblock In {\em Seminars in cancer biology}, volume~20, pages 222--233.
  Elsevier, 2010.

\bibitem{Chitturi2015}
Bhadrachalam Chitturi.
\newblock Tighter upper bound for sorting permutations with prefix
  transpositions.
\newblock {\em Theoretical Computer Science}, 602:22--31, 2015.

\bibitem{chitturi200918}
Bhadrachalam Chitturi, William Fahle, Zhaobing Meng, Linda Morales, Charles~O
  Shields, Ivan~Hal Sudborough, and Walter Voit.
\newblock An (18/11) n upper bound for sorting by prefix reversals.
\newblock {\em Theoretical Computer Science}, 410(36):3372--3390, 2009.

\bibitem{Chitturi2008}
Bhadrachalam Chitturi and I~Hal Sudborough.
\newblock Bounding prefix transposition distance for strings and permutations.
\newblock In {\em Proceedings of the 41st Annual Hawaii International
  Conference on System Sciences (HICSS 2008)}, pages 468--468. IEEE, 2008.

\bibitem{Chitturi2012}
Bhadrachalam Chitturi and I~Hal Sudborough.
\newblock Bounding prefix transposition distance for strings and permutations.
\newblock {\em Theoretical Computer Science}, 421:15--24, 2012.

\bibitem{Christie1998}
David~Alan Christie.
\newblock {\em Genome rearrangement problems}.
\newblock PhD thesis, University of Glasgow, 1998.

\bibitem{Dias2002}
Zanoni Dias and Joao Meidanis.
\newblock Sorting by prefix transpositions.
\newblock In {\em International Symposium on String Processing and Information
  Retrieval}, pages 65--76. Springer, 2002.

\bibitem{Dobzhansky1938}
Th~Dobzhansky and Alfred~H Sturtevant.
\newblock Inversions in the chromosomes of drosophila pseudoobscura.
\newblock {\em Genetics}, 23(1):28, 1938.

\bibitem{Elias2006}
Isaac Elias and Tzvika Hartman.
\newblock A 1.375-approximation algorithm for sorting by transpositions.
\newblock {\em IEEE/ACM Transactions on Computational Biology and
  Bioinformatics}, 3(4):369--379, 2006.

\bibitem{eriksson2001sorting}
Henrik Eriksson, Kimmo Eriksson, Johan Karlander, Lars Svensson, and Johan
  W{\"a}stlund.
\newblock Sorting a bridge hand.
\newblock {\em Discrete Mathematics}, 241(1-3):289--300, 2001.

\bibitem{Feng2010}
Xuerong Feng, Bhadrachalam Chitturi, and Hal Sudborough.
\newblock Sorting circular permutations by bounded transpositions.
\newblock In {\em Advances in Computational Biology}, pages 725--736. Springer,
  2010.

\bibitem{Fortuna2005}
Vinicius~Jos Fortuna.
\newblock Dist{\^a}ncias de transposi{\c{c}}ao entre genomas.
\newblock Master's thesis, Institute of Computing, University of Campinas,
  2005.

\bibitem{gates1979bounds}
William~H Gates and Christos~H Papadimitriou.
\newblock Bounds for sorting by prefix reversal.
\newblock {\em Discrete mathematics}, 27(1):47--57, 1979.

\bibitem{hannenhalli1999transforming}
Sridhar Hannenhalli and Pavel~A Pevzner.
\newblock Transforming cabbage into turnip: polynomial algorithm for sorting
  signed permutations by reversals.
\newblock {\em Journal of the ACM (JACM)}, 46(1):1--27, 1999.

\bibitem{heath1998sorting}
Lenwood~S Heath and John Paul~C Vergara.
\newblock Sorting by bounded block-moves.
\newblock {\em Discrete Applied Mathematics}, 88(1-3):181--206, 1998.

\bibitem{Jones2004}
Neil~C Jones, Pavel~A Pevzner, and Pavel Pevzner.
\newblock {\em An introduction to bioinformatics algorithms}.
\newblock MIT press, 2004.

\bibitem{kantar2017genetics}
Michael~B Kantar, Amber~R Nashoba, Justin~E Anderson, Benjamin~K Blackman, and
  Loren~H Rieseberg.
\newblock The genetics and genomics of plant domestication.
\newblock {\em Bioscience}, 67(11):971--982, 2017.

\bibitem{kaplan2000faster}
Haim Kaplan, Ron Shamir, and Robert~E Tarjan.
\newblock A faster and simpler algorithm for sorting signed permutations by
  reversals.
\newblock {\em SIAM Journal on Computing}, 29(3):880--892, 2000.

\bibitem{kececioglu1995exact}
John Kececioglu and David Sankoff.
\newblock Exact and approximation algorithms for sorting by reversals, with
  application to genome rearrangement.
\newblock {\em Algorithmica}, 13(1):180--210, 1995.

\bibitem{labarre2008edit}
Anthony Labarre.
\newblock Edit distances and factorisations of even permutations.
\newblock In {\em European Symposium on Algorithms}, pages 635--646. Springer,
  2008.

\bibitem{labarre2020sorting}
Anthony Labarre.
\newblock Sorting by prefix block-interchanges.
\newblock In {\em 31st International Symposium on Algorithms and Computation
  (ISAAC 2020)}. Schloss Dagstuhl-Leibniz-Zentrum f{\"u}r Informatik, 2020.

\bibitem{lakshmivarahan1993symmetry}
Sivaramakrishnan Lakshmivarahan, Jung-Sing Jwo, and Sudarshan~K. Dhall.
\newblock Symmetry in interconnection networks based on cayley graphs of
  permutation groups: A survey.
\newblock {\em Parallel Computing}, 19(4):361--407, 1993.

\bibitem{Li2006}
Zimao Li, Lusheng Wang, and Kaizhong Zhang.
\newblock Algorithmic approaches for genome rearrangement: a review.
\newblock {\em IEEE Transactions on Systems, Man, and Cybernetics, Part C
  (Applications and Reviews)}, 36(5):636--648, 2006.

\bibitem{lin2005efficient}
Ying~Chih Lin, Chin~Lung Lu, Hwan-You Chang, and Chuan~Yi Tang.
\newblock An efficient algorithm for sorting by block-interchanges and its
  application to the evolution of vibrio species.
\newblock {\em Journal of Computational Biology}, 12(1):102--112, 2005.

\bibitem{nair2020improved}
Pramod~P Nair and Rajan Sundaravaradhan.
\newblock An improved upper bound for genome rearrangement by prefix
  transpositions.
\newblock In {\em IEEE Conference on Advanced Computing and Communication
  Technologies for High Performance Applications}. IEEE, 2020.

\bibitem{nair2020new}
Pramod~P Nair, Rajan Sundaravaradhan, and Bhadrachalam Chitturi.
\newblock A new upper bound for sorting permutations with prefix
  transpositions.
\newblock {\em Discrete Mathematics, Algorithms and Applications}, 2020.

\bibitem{nair2021}
Pramod~P Nair, Rajan Sundaravaradhan, and Bhadrachalam Chitturi.
\newblock Improved upper bound for sorting permutations by prefix
  transpositions.
\newblock {\em Theoretical Computer Science}, 896:158--167, 2021.

\bibitem{Oliveira2019}
Andre~Rodrigues Oliveira, Klairton~Lima Brito, Ulisses Dias, and Zanoni Dias.
\newblock On the complexity of sorting by reversals and transpositions
  problems.
\newblock {\em Journal of Computational Biology}, 26(11):1223--1229, 2019.

\bibitem{sanger1977dna}
Frederick Sanger, Steven Nicklen, and Alan~R Coulson.
\newblock Dna sequencing with chain-terminating inhibitors.
\newblock {\em Proceedings of the national academy of sciences},
  74(12):5463--5467, 1977.

\bibitem{shendure2017dna}
Jay Shendure, Shankar Balasubramanian, George~M Church, Walter Gilbert, Jane
  Rogers, Jeffery~A Schloss, and Robert~H Waterston.
\newblock Dna sequencing at 40: past, present and future.
\newblock {\em Nature}, 550(7676):345--353, 2017.

\bibitem{Watterson1982}
GA~Watterson, Warren~J Ewens, Thomas~Eric Hall, and A~Morgan.
\newblock The chromosome inversion problem.
\newblock {\em Journal of Theoretical Biology}, 99(1):1--7, 1982.

\end{thebibliography}
\bibliographystyle{plain}

\chapter*{List of publications based on the research work} \addcontentsline{toc}{chapter}{\hspace*{.53cm}List of publications based on the research work}
{\Large\textbf{International Journal }}
\begin{enumerate}
\item Nair, P.P., Sundaravaradhan, R., Chitturi, B., A new upper bound for sorting permutations with prefix transpositions, \emph{Discrete Mathematics, Algorithms and Applications}, 12(6):2050077, 2020, .
\item Nair, P.P., Sundaravaradhan, R., Chitturi, B., Improved upper bound for sorting permutations with prefix transpositions, \emph{Theoretical Computer Science}, 896:158-167, 2021.
\end{enumerate}
{\Large\textbf{International Conference}}
\begin{enumerate}
\item[3.] Nair, P.P., Sundaravaradhan, R., An improved upper bound for genome rearrangement by prefix transpositions, \emph{Proceedings - 2020 Advanced Computing and Communication Technologies for High Performance Applications, ACCTHPA 2020}, 115–119, 2020.
\end{enumerate}
\end{document}